\documentclass[draft,journal,a4paper,oneside,onecolumn]{IEEEtran}
\usepackage{amsmath,amsfonts, amssymb}

\newtheorem{thm}{Theorem}

\newtheorem{lem}{Lemma}
\newtheorem{df}{Definition}
\newtheorem{rem}{Remark}

\newcommand{\GF}{\mathrm{GF}}
\newcommand{\GFq}{\mathrm{GF}(q)}

\newcommand{\hB}{\widehat{B}}

\newcommand{\hcB}{\widehat{\mathcal{B}}}
\newcommand{\A}{\mathcal{A}}
\newcommand{\B}{\mathcal{B}}
\newcommand{\C}{\mathcal{C}}

\newcommand{\cS}{\mathcal{S}}
\newcommand{\T}{\mathcal{T}}
\newcommand{\U}{\mathcal{U}}
\newcommand{\bU}{\overline{\mathcal{U}}}
\newcommand{\V}{\mathcal{V}}

\newcommand{\X}{\mathcal{X}}
\newcommand{\hcX}{\widehat{\mathcal{X}}}
\newcommand{\tcX}{\widetilde{\mathcal{X}}}
\newcommand{\tX}{\widetilde{X}}
\newcommand{\Y}{\mathcal{Y}}
\newcommand{\Z}{\mathcal{Z}}

\newcommand{\G}{\mathcal{G}}
\newcommand{\bcA}{\boldsymbol{\mathcal{A}}}
\newcommand{\bcB}{\boldsymbol{\mathcal{B}}}
\newcommand{\bchB}{\boldsymbol{\widehat{\mathcal{B}}}}
\newcommand{\aalpha}{\boldsymbol{\alpha}}
\newcommand{\bbeta}{\boldsymbol{\beta}}
\newcommand{\kkappa}{\boldsymbol{\kappa}}

\newcommand{\cc}{\boldsymbol{c}}
\newcommand{\mm}{\boldsymbol{m}}

\newcommand{\bp}{\boldsymbol{p}}

\newcommand{\uu}{\boldsymbol{u}}
\newcommand{\vv}{\boldsymbol{v}}
\newcommand{\ww}{\boldsymbol{w}}
\newcommand{\xx}{\boldsymbol{x}}
\newcommand{\yy}{\boldsymbol{y}}
\newcommand{\zz}{\boldsymbol{z}}

\newcommand{\txx}{\widetilde{\boldsymbol{x}}}
\newcommand{\tzz}{\widetilde{\boldsymbol{z}}}
\newcommand{\hxx}{\widehat{\boldsymbol{x}}}

\newcommand{\hx}{\widehat{x}}
\newcommand{\hX}{\widehat{X}}
\newcommand{\hY}{\widehat{Y}}

\newcommand{\tg}{\widetilde{g}}
\newcommand{\tZ}{\widetilde{Z}}
\newcommand{\tz}{\widetilde{z}}
\newcommand{\tcZ}{\widetilde{\mathcal{Z}}}
\newcommand{\FF}{\boldsymbol{F}}

\newcommand{\e}{\varepsilon}

\newcommand{\lrB}[1]{\left[{#1}\right]}
\newcommand{\lrb}[1]{\left\{{#1}\right\}}
\newcommand{\lrsb}[1]{\left({#1}\right)}

\newcommand{\Error}{\mathrm{Error}}
\newcommand{\zero}{\boldsymbol{0}}
\newcommand{\one}{\boldsymbol{1}}
\newcommand{\limn}{\lim_{n\to\infty}}

\newcommand{\Encoder}{\varphi}
\newcommand{\Decoder}{\varphi^{-1}}
\newcommand{\Security}{\mathrm{Leakage}}
\newcommand{\Capacity}{\mathrm{Capacity}}
\newcommand{\Rate}{\mathrm{Rate}}
\newcommand{\im}{\mathrm{Im}}
\newcommand{\qed}{$\blacksquare$}
\newcommand{\lA}{l_{\A}}
\newcommand{\lB}{l_{\B}}
\newcommand{\lhB}{l_{\hcB}}
\newcommand{\eA}{\e_{\A}}

\newcommand{\ehB}{\e_{\hcB}}
\newcommand{\RA}{R_{\A}}
\newcommand{\RB}{R_{\B}}
\newcommand{\RhB}{R_{\hcB}}
\newcommand{\markov}{\leftrightarrow}

\allowdisplaybreaks

\title{
Construction of Codes
for Wiretap Channel and Secret Key Agreement
from Correlated Source Outputs
by Using Sparse Matrices
}
\author{
Jun~Muramatsu~\IEEEmembership{Member,~IEEE,}
and~Shigeki Miyake~\IEEEmembership{Member,~IEEE,}
 \thanks{J.~Muramatsu is with
        NTT Communication Science Laboratories, NTT Corporation,
        2-4, Hikaridai, Seika-cho, Soraku-gun, Kyoto 619-0237, Japan
        (E-mail: pure@cslab.kecl.ntt.co.jp).
        S.~Miyake is with
        NTT Network Innovation Laboratories, NTT Corporation,
        Midori-cho 3-9-11, Musashino-shi, Tokyo 180-8585, Japan
        (E-mail: miyake.shigeki@lab.ntt.co.jp).
	}
  \thanks{This paper was presented in part at
  ``Construction of wiretap channel codes by using sparse matrices,''
  {\em Proc. 2009 IEEE Information Theory Workshop (ITW2009)}, Taormina,
  Italy, pp.\ 105--109, 2009.
  This paper is submitted to
  {\it IEEE Transactions on Information Theory}, Feb.\ 2010.
  }
}
  
\begin{document}
\maketitle

\begin{abstract}
 The aim of this paper is to prove
 coding theorems for the wiretap channel
 and secret key agreement
 based on the the notion of a hash property
 for an ensemble of functions.
 These theorems imply that codes using sparse
 matrices can achieve the optimal rate.
 Furthermore, fixed-rate universal coding theorems for a wiretap channel
 and a secret key agreement are also proved.
\end{abstract}

\section{Introduction}
The aim of this paper is to prove the coding theorems
for the wiretap channel (Fig.\ \ref{fig:wiretap}) introduced in \cite{W75b}
and secret key agreement problem (Fig.\ \ref{fig:ska})
introduced in \cite{Mau93}\cite{AC93}.
The proof of theorems is based on the notion of a hash property for an
ensemble of functions introduced in~\cite{HASH}\cite{HASH-UNIV}.
This notion provides a sufficient condition for the achievability
of coding theorems.
Since an ensemble of sparse matrices has a hash property,
we can construct codes by using sparse matrices
where the rate of codes is close to the optimal rate.
In the construction of codes, we use minimum-divergence encoding,
maximum-likelihood decoding, and minimum-entropy decoding,
where we can use the approximation methods introduced in
\cite{FWK05}\cite{CME05} to realize these operations.

Wiretap channel coding using a sparse matrices
is studied in \cite{TDCMM07} for binary erasure wiretap channels.
On the other hand, our construction can be applied
to any stationary memoryless channel.
It should be noted here that the encoder design is based on
the standard channel code presented in
\cite{SWLDPC}\cite{HASH}\cite{HASH-UNIV}\cite{MM08}.
Furthermore, we prove the fixed-rate universal coding theorem for a wiretap
channel, where our construction is reliable and secure
for any channel under some conditions specified by the encoding
rate. Universality is not considered in \cite{W75b}\cite{TDCMM07}.

The secret key agreement from correlated source outputs
using sparse matrices is studied in \cite{ISIT04}\cite{CRYPTLDPC},
where both non-universal and universal codes are considered.
Our construction is the same as that proposed in \cite{CRYPTLDPC}.
It should be noted that the linearity of functions is not assumed
in our proof of reliability and security
while it is assumed in \cite{CRYPTLDPC}.
Furthermore, an expurgated ensemble of sparse matrices
is not assumed in our proof while it is assumed in \cite{CRYPTLDPC}.

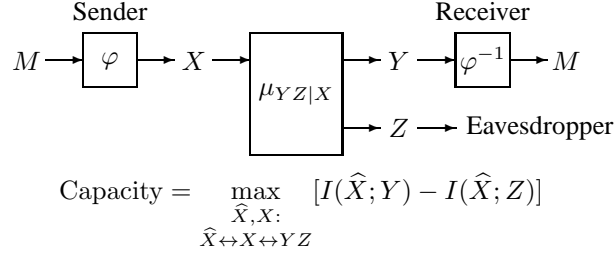
\begin{figure}
 \begin{center}
  \unitlength 0.5mm
  \begin{picture}(152,60)(0,-10)
   \put(5,35){\makebox(0,0){$M$}}
   \put(10,35){\vector(1,0){10}}
   \put(27,48){\makebox(0,0){Sender}}
   \put(20,28){\framebox(14,14){$\Encoder$}}
   \put(34,35){\vector(1,0){10}}
   \put(49,35){\makebox(0,0){$X$}}
   \put(54,35){\vector(1,0){10}}
   \put(64,10){\framebox(24,32){$\mu_{YZ|X}$}}
   \put(88,35){\vector(1,0){10}}
   \put(88,17){\vector(1,0){10}}
   \put(103,35){\makebox(0,0){$Y$}}
   \put(103,17){\makebox(0,0){$Z$}}
   \put(108,35){\vector(1,0){10}}
   \put(108,17){\vector(1,0){10}}
   \put(140,17){\makebox(0,0){Eavesdropper}}
   \put(125,48){\makebox(0,0){Receiver}}
   \put(118,28){\framebox(14,14){$\Decoder$}}
   \put(132,35){\vector(1,0){10}}
   \put(147,35){\makebox(0,0){$M$}}
   \put(76,-5){\makebox(0,0){
   $\Capacity=\displaystyle{\max_{\substack{\hX,X:\\\hX\markov X\markov YZ}}}[I(\hX;Y)-I(\hX;Z)]$}}
  \end{picture}
 \end{center}
 \caption{Wiretap Channel Coding}
 \label{fig:wiretap}
\end{figure}

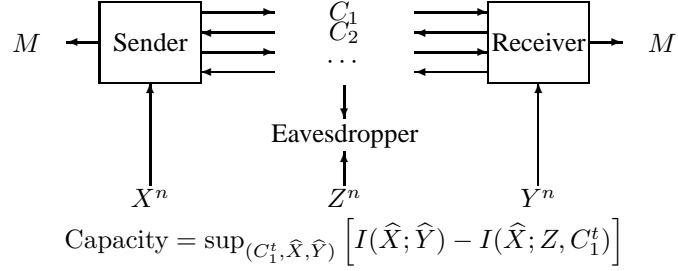
\begin{figure}
 \begin{center}
  \unitlength 0.44mm
  \begin{picture}(170,80)(5,-35)
   \put(-5,30){\makebox(0,0){$M$}}
   \put(17,30){\vector(-1,0){10}}
   \put(17,18){\framebox(30,24){Sender}}
   \put(47,39){\vector(1,0){22}}
   \put(90,39){\makebox(0,0){$C_1$}}
   \put(111,39){\vector(1,0){22}}
   \put(69,33){\vector(-1,0){22}}
   \put(90,33){\makebox(0,0){$C_2$}}
   \put(133,33){\vector(-1,0){22}}
   \put(47,27){\vector(1,0){22}}
   \put(90,24){\makebox(0,0){$\cdots$}}
   \put(111,27){\vector(1,0){22}}
   \put(69,21){\vector(-1,0){22}}
   \put(133,21){\vector(-1,0){22}}
   \put(133,18){\framebox(30,24){Receiver}}
   \put(163,30){\vector(1,0){10}}
   \put(185,30){\makebox(0,0){$M$}}
   \put(32,-13){\vector(0,1){31}}
   \put(32,-17){\makebox(0,0){$X^n$}}
   \put(148,-13){\vector(0,1){31}}
   \put(148,-17){\makebox(0,0){$Y^n$}}
   \put(90,17){\vector(0,-1){10}}
   \put(90,2){\makebox(0,0){Eavesdropper}}
   \put(90,-13){\vector(0,1){10}}
   \put(90,-17){\makebox(0,0){$Z^n$}}
   \put(90,-30){\makebox(0,0){$\Capacity=\sup_{(C_1^t,\hX,\hY)}\lrB{I(\hX;\hY)-I(\hX;Z,C_1^t)}$}}
  \end{picture}
 \end{center}
 \caption{Secret Key Agreement from Correlated Source Outputs}
 \label{fig:ska}
\end{figure}

\section{Definitions and Notations}
Throughout this paper, we use the following definitions and notations.

The cardinality of a set $\U$ is denoted by $|\U|$,
$\U^c$ denotes the compliment of $\U$,
and $\U\setminus\V\equiv\U\cap\V^c$.

Column vectors and sequences are denoted in boldface.
Let $A\uu$ denote a value taken by a function $A:\U^n\to\bU$ at
$\uu\equiv(u_1,\ldots,u_n)\in\U^n$,
where $\U^n$ is a domain of the function.
It should be noted that $A$ may be nonlinear.
When $A$ is a linear function expressed by an $l\times n$ matrix,
we assume that $\U$ is a finite field and the range of functions is
defined by $\bU\equiv\U^l$.
It should be noted that this assumption is not essential
for general (nonlinear) functions
because discussion is not changed if $l\log|\U|$ is replaced by $\log|\bU|$.
For a set $\A$ of functions, let $\im \A$ be defined as
\begin{align*}
 \im\A &\equiv \bigcup_{A\in\A}\{A\uu: \uu\in\U^n\}.
\end{align*}
We define sets $\C_A(\cc)$, $\C_{AB}(\cc,\mm)$,
and $\C_{AB\hB}(\cc,\mm,\ww)$ as
\begin{align*}
 \C_A(\cc)
 &\equiv\{\uu: A\uu = \cc\}
 \\
 \C_{AB}(\cc,\mm)
 &\equiv\{\uu: A\uu = \cc, B\uu = \mm\}
 \\
 \C_{AB\hB}(\cc,\mm,\ww)
 &\equiv\{\uu: A\uu = \cc, B\uu = \mm, \hB\uu = \ww\}.
\end{align*}
In the context of linear codes, $\C_A(\cc)$ is call a coset determined
by $\cc$.

Let $p$ and $p'$ be probability distributions
and let $q$ and $q'$ be conditional probability distributions.
Then entropy $H(p)$, conditional entropy $H(q|p)$,
divergence $D(p\|p')$, and conditional divergence $D(q\|q'|p)$
are defined as
\begin{align*}
 H(p)
 &\equiv\sum_{u}p(u)\log\frac 1{p(u)}
 \\
 H(q|p)
 &\equiv\sum_{u,v}q(u|v)p(v)\log\frac 1{q(u|v)}
 \\
 D(p\parallel p')
 &\equiv \sum_{u}p(u) \log\frac{p(u)}{p'(u)}
 \\
 D(q\parallel q' | p)
 &\equiv \sum_{v} p(v)\sum_{u}q(u|v) \log\frac{q(u|v)}{q'(u|v)},
\end{align*}
where we assume the base $2$ of the logarithm.

Let $\mu_{UV}$ be the joint probability distribution of random variables
$U$ and $V$.
Let  $\mu_{U}$ and $\mu_{V}$ be the respective marginal distributions
and $\mu_{U|V}$ be the conditional probability distribution.
Then the entropy $H(U)$, the conditional entropy $H(U|V)$, and the mutual
information $I(U;V)$ of random variables are defined as
\begin{align*}
 H(U)
 &\equiv H(\mu_U)
 \\
 H(U|V)
 &\equiv H(\mu_{U|V}|\mu_{V})
 \\
 I(U;V)
 &\equiv H(U)-H(U|V).
\end{align*}

Let $\nu_{\uu}$ and $\nu_{\uu|\vv}$ be defined as
\begin{align*}
 \nu_{\uu}(u)
 &\equiv
 \frac {|\{1\leq i\leq n : u_{i}=u\}|}n
 \\
 \nu_{\uu|\vv}(u|v)
 &\equiv \frac{\nu_{\uu\vv}(u,v)}{\nu_{\vv}(v)}.
\end{align*}
We call $\nu_{\uu}$ a type of $\uu\in\U^n$
and $\nu_{\uu|\vv}$ a conditional type.
Let $U\equiv\nu_U$ be the type of a sequence
and $U|V\equiv\nu_{U|V}$ be the conditional type of a sequence
given a sequence of type $U$.
Then a set of typical sequences $\T_{U}$
and a set of conditionally typical sequences $\T_{U|V}(\vv)$
are defined as
\begin{align*}
 \T_{U}
 &\equiv
 \lrb{\uu:
 \nu_{\uu}=\nu_U
 }
 \\
 \T_{U|V}(\vv)
 &\equiv
 \lrb{\uu:
 \nu_{\uu|\vv}=\nu_{U|V}
 }.
\end{align*}
The empirical entropy, the empirical conditional entropy, and
empirical mutual information are defined as
\begin{align*}
 H(\uu)
 &\equiv H(\nu_{\uu})
 \\
 H(\uu|\vv)
 &\equiv H(\nu_{\uu|\vv}|\nu_{\vv})
 \\
 I(\uu;\vv)
 &\equiv H(\uu)-H(\uu|\vv).
\end{align*}
A set of typical sequences $\T_{U,\gamma}$
and a set of conditionally typical sequences $\T_{U|V,\gamma}(\vv)$
are defined as
\begin{align*}
 \T_{U,\gamma}
 &\equiv
 \lrb{\uu:
 D(\nu_{\uu}\|\mu_{U})<\gamma
 }
 \\
 \T_{U|V,\gamma}(\vv)
 &\equiv
 \lrb{\uu:
 D(\nu_{\uu|\vv}\|\mu_{U|V}|\nu_{\vv})<\gamma
 }.
\end{align*}
We use several lemmas for the method of the types described in Appendix.

In the construction of codes, we use a minimum-divergence encoder
\begin{align}
 \tg_{AB\hB}(\cc,\mm,\ww)
 &\equiv\arg\min_{\xx'\in\C_{AB\hB}(\cc,\mm,\ww)}D(\nu_{\xx'}\|\mu_X),
 \label{eq:md}
\end{align}
a maximum-likelihood decoder
\begin{align}
 g_A(\cc|\yy)&\equiv\arg\max_{\xx'\in\C_A(\cc)}\mu_{X|Y}(\xx'|\yy),
 \label{eq:ml}
\end{align}
and a minimum-entropy decoder
\begin{align}
 \tg_{A}(\cc|\yy)
 &\equiv\arg\min_{\xx'\in\C_{A}(\cc)}H(\xx'|\yy).
 \label{eq:me}
\end{align}
The minimum-divergence encoder
assigns a message to a typical sequence as close as
possible to the input distribution, where the typical sequence is in the
coset determined by $\cc$.
The time complexity of encoding and decoding is exponential with respect
to the block length by using the exhaustive search.
It should be noted that
the linear programming method introduced \cite{FWK05}
and \cite{CME05} can be applied to these encoder and decoders
by assuming that $\X=\Y=\GF(2)$ and $A$, $B$, and $\hB$ are linear
functions,
where the linear programming method may not find the integral solution.
Details are described in Section~\ref{sec:binary}.
It should be noted here that
we do not discuss the performance of the linear programming methods
in this paper.

We define $\chi(\cdot)$ as
\begin{align*}
 \chi(a = b)
 &\equiv
 \begin{cases}
  1,&\text{if}\ a = b
  \\
  0,&\text{if}\ a\neq b
 \end{cases}
 \\
 \chi(a \neq b)
 &\equiv
 \begin{cases}
  1,&\text{if}\ a \neq b
  \\
  0,&\text{if}\ a = b.
 \end{cases}
\end{align*}

Finally, we use the following definitions in Appendix.
For $\gamma,\gamma'>0$, we define
\begin{align}
 \lambda_{\U}
 &\equiv \frac{|\U|\log[n+1]}n
 \label{eq:lambda}
 \\
 \zeta_{\U}(\gamma)
 &\equiv
 \gamma-\sqrt{2\gamma}\log\frac{\sqrt{2\gamma}}{|\U|}
 \label{eq:zeta}
 \\
 \zeta_{\U|\V}(\gamma'|\gamma)
 &\equiv
 \gamma'-\sqrt{2\gamma'}\log\frac{\sqrt{2\gamma'}}{|\U||\V|}
 +\sqrt{2\gamma}\log|\U|
 \label{eq:zetac}
 \\
 \eta_{\U}(\gamma)
 &\equiv
 -\sqrt{2\gamma}\log\frac{\sqrt{2\gamma}}{|\U|}
 +\frac{|\U|\log[n+1]}n
 \label{eq:def-eta}
\end{align}
It should be noted here that
the product set $\U\times\V$ is denoted by $\U\V$
when it appears in the subscript of these functions.
We define $h(\theta)$ for $0\leq \theta\leq 1$ as
\begin{equation}
 h(\theta)\equiv -\theta\log\theta-[1-\theta]\log(1-\theta).
  \label{eq:def-h}
\end{equation}
We define $|\cdot|^+$ as
\begin{equation}
 |\theta|^+
  \equiv
  \begin{cases}
   \theta,&\text{if}\ \theta>0,\\
   0,&\text{if}\ \theta\leq 0.
  \end{cases}
  \label{eq:def-plus}
\end{equation}

\section{{$(\aalpha,\bbeta)$-hash Property}}

In the following, we review the notion of the hash property for an
ensemble of functions, which is introduced in \cite{HASH}.
This provides a sufficient condition
for coding theorems, where the linearity of functions is not assumed.
We prove coding theorems based on this notion.

\begin{df}[\cite{HASH}]
 Let $\bcA\equiv\{\A_n\}_{n=1}^{\infty}$
 be a sequence of sets such that
 $\A_{n}$ is a set of functions $A:\U^n\to\bU_{\A_{n}}$
 satisfying
 \begin{equation}
  \limn \frac{\log\frac{|\bU_{\A_n}|}{|\im\A_n|}}n=0.
   \tag{H1}
   \label{eq:imA}
 \end{equation}
 For a probability distribution $p_{A,n}$ on $\A_n$, we
 call a sequence $(\bcA,\bp_{A})\equiv\{(\A_n,p_{A,n})\}_{n=1}^{\infty}$
 an {\em ensemble}.
 Then, $(\bcA,\bp_{A})$ has an
 $(\aalpha_{A},\bbeta_{A})$-{\em hash property} if
 there are two sequences
 $\aalpha_{A}\equiv\{\alpha_{A}(n)\}_{n=1}^{\infty}$ and
 $\bbeta_{A}\equiv\{\beta_{A}(n)\}_{n=1}^{\infty}$
 such that
 \begin{align}
  &\limn \alpha_{A}(n)=1
  \tag{H2}
  \label{eq:alpha}
  \\
  &\limn \beta_{A}(n)=0
  \tag{H3}
  \label{eq:beta}
 \end{align}
 and
 \begin{equation}
  \sum_{\substack{
   \uu\in\T
  \\
  \uu'\in\T'
  }}
  p_{A,n}\lrsb{\lrb{A: A\uu = A\uu'}}
  \leq
  |\T\cap\T'|
  +\frac{|\T||\T'|\alpha_{A}(n)}{|\im\A_n|}
  +\min\{|\T|,|\T'|\}\beta_{A}(n)
  \tag{H4}
  \label{eq:hash}
 \end{equation}
 for any $\T,\T'\subset\U^n$.
 Throughout this paper,
 we omit dependence of $\A$, $p_{A}$, $\alpha_{A}$ and $\beta_{A}$ on $n$.
\end{df}

In the following, we present two examples of ensembles that have a hash
property.

\noindent{\bf Example 1:}
In this example, we consider
a universal class of hash functions introduced in \cite{CW}.
A set $\A$ of functions $A:\U^n\to\bU_{\A}$ is called
a {\em universal class of hash functions} if 
\[
 |\lrb{A: A\uu=A\uu'}|\leq \frac{|\A|}{|\bU_{\A}|}
\]
for any $\uu\neq\uu'$.
For example, the set of all functions on $\U^n$
and the set of all linear functions $A:\U^n\to\U^{\lA}$
are classes of universal hash functions (see \cite{CW}).
When $\A$ is a universal class of hash functions
and  $p_A$ is the uniform distribution on $\A$,
we have
\begin{align*}
 \sum_{\substack{
 \uu\in\T
 \\
 \uu'\in\T'
 }}
 p_A\lrsb{\lrb{A: A\uu=A\uu'}}
 \leq
 |\T\cap\T'|+\frac{|\T||\T'|}{|\im\A|}.
\end{align*}
This implies that $(\bcA,\bp_A)$ has a $(\one,\zero)$-hash property,
where $\one(n)\equiv 1$ and $\zero(n)\equiv 0$ for every $n$.

\noindent{\bf Example 2:}
In this example, we consider a set of linear functions
$A:\U^n\to\U^{\lA}$.
It was discussed in the above example that
the uniform distribution on
the set of all linear functions has a $(\one,\zero)$-hash property.
In the following, we introduce the ensemble of
$q$-ary sparse matrices proposed in \cite{HASH}.
Let $\U\equiv\GFq$ and $\lA\equiv nR$ for given $0<R<1$.
We generate an $\lA\times n$ matrix $A$ with the following procedure,
where at most $\tau$ random nonzero elements are introduced
in every row.
\begin{enumerate}
\item Start from an all-zero matrix.
\item For each $i\in\{1,\ldots,n\}$, repeat the following
      procedure $\tau$ times:
      \begin{enumerate}
       \item Choose $(j,a)\in\{1,\ldots,\lA\}\times\GFq$ uniformly at random.
       \item Add  $a$ to the  $(j,i)$ component of $A$.
      \end{enumerate}
\end{enumerate}
Let $(\bcA,\bp_{A})$ be an ensemble corresponding to the above procedure,
where $\tau=O(\log \lA)$ is even.
It is proved in \cite[Theorem 2]{HASH}
that there is $(\aalpha_A,\bbeta_A)$ such that
$(\bcA,\bp_A)$ has an $(\aalpha_A,\bbeta_A)$-hash property.

In the following,
let $\A$ be a set of functions $A:\U^n\to\bU_{\A}$
and assume that
$p_C$ is the uniform distribution on $\im\A$,
and random variables $A$ and $C$ are mutually independent, that is,
\begin{align*}
 p_C(\cc)
 &=
 \begin{cases}
  \frac 1{|\im\A|},&\text{if}\ \cc\in\im\A
  \\
  0,&\text{if}\ \cc\in\bU_{\A}\setminus\im\A
 \end{cases}
 \\
 p_{AC}(A,\cc)
 &=p_A(A)p_C(\cc)
\end{align*}
for any $A$ and $\cc$.
We have the following lemmas,
where it is not necessary to assume the linearity of functions.

\begin{lem}[{\cite[Lemma 1]{HASH}}]
 \label{lem:Anotempty}
 If $(\A,p_A)$ satisfies (\ref{eq:hash}), then
 \[
 p_A\lrsb{\lrb{
 A: \lrB{\G\setminus\{\uu\}}\cap\C_A(A\uu)\neq \emptyset
 }}
 \leq 
 \frac{|\G|\alpha_A}{|\im\A|} + \beta_A
 \]
 for all $\G\subset\U^n$ and all $\uu\in\U^n$.
\end{lem}

\begin{lem}[{\cite[Lemma 2]{HASH}}]
\label{lem:saturating}
 If $(\A,p_A)$ satisfies (\ref{eq:hash}), then
 \begin{equation*}
  p_{AC}\lrsb{\lrb{(A,\cc):
  \T\cap\C_{A}(\cc)=\emptyset
  }}
  \leq
  \alpha_{A}-1+\frac{|\im\A|\lrB{\beta_{A}+1}}{|\T|}
 \end{equation*}
 for all $\T\neq\emptyset$.
\end{lem}

Finally, we consider the independent joint ensemble $p_{AB}$
of linear matrices.
The following lemma asserts that it is sufficient to
assume the hash property of $(\bcA,\bp_A)$ and 
$(\bcB,\bp_B)$ to satisfy the hash property of
$(\bcA\times\bcB,\bp_{AB})$ when they are ensembles of linear matrices.
\begin{lem}[{\cite[Lemma 7]{HASH}}]
 \label{lem:hash-linApB}
 For two ensembles $(\bcA,\bp_A)$ and
 $(\bcB,\bp_B)$, of $\lA\times n$ and $\lB\times n$ linear matrices,
 respectively,
 let $p_{AB}$ be the joint distribution defined as
 \[
 p_{AB}(A,B)\equiv p_A(A)p_B(B).
 \]
 Then $(\bcA\times\bcB,\bp_{AB})$ has an
 $(\aalpha_{AB},\bbeta_{AB})$-hash property
 for the ensemble of functions $A\oplus B:\U^n\to\U^{\lA+\lB}$ defined as
 \[
 A\oplus B(\uu)\equiv(A\uu,B\uu),
 \]
 where
 \begin{align*}
  \alpha_{AB}(n)
  &=\alpha_A(n)\alpha_B(n)
  \\
  \beta_{AB}(n)
  &=\beta_A(n)+\beta_B(n).
 \end{align*}
\end{lem}

\section{Wiretap Channel Coding}

In this section we consider 
the wiretap channel coding problem illustrated in
Fig.\ \ref{fig:wiretap}, where no common message and perfect secrecy are
assumed.
A wiretap channel is characterized by
the conditional probability distribution $\mu_{YZ|X}$,
where $X$, $Y$, and $Z$ are random variables corresponding to
the channel input of a sender, the channel output of a legitimate receiver
and the channel output of an eavesdropper.
Then the capacity\footnote{
It is stated in \cite{ITW09} that
the auxiliary random variable can be eliminated
by applying \cite[Theorem 7]{Hayashi06} and \cite[Theorem 3]{KS05}.
In fact, because of the authors misunderstanding about the result of
\cite[Theorem 3]{KS05}, the statement of \cite{ITW09} may not be true.
They wish to thank Prof. Shamai (Shitz),
Prof. Oohama, and Prof. Koga, for helpful discussions.
}
of this channel is derived in \cite[Eq.~(11)]{CK78} as
\begin{equation}
 \Capacity\equiv\max_{\substack{
  \hX,X:\\
 \hX\markov X\markov YZ
  }}
  \lrB{I(\hX;Y)-I(\hX;Z)},
  \label{eq:wiretap-capacity}
\end{equation}
where the maximum is taken over all probability distribution
$\mu_{\hX X}$ and the joint distribution $\mu_{\hX XYZ}$
is given by
\begin{equation}
 \mu_{\hX XYZ}(\hx,x,y,z)
  \equiv\mu_{YZ|X}(y,z|x)\mu_{\hX X}(\hx,x).
  \label{eq:markov-wiretap}
\end{equation}
If a channel between $X$ and $Y$ is more capable than a channel
between $X$ and $Z$, that is,
\[
 I(X;Y)\geq I(X;Z)
\]
is satisfied for every input $X$,
then the capacity of this channel is simplified as
\begin{equation}
 \Capacity\equiv\max_{X}\lrB{I(X;Y)-I(X;Z)},
  \label{eq:wiretap-capacity-capable}
\end{equation}
where the maximum is taken over all random variables $X$
and the joint distribution of random variable $(X,Y,Z)$
is given by
\begin{align}
 \mu_{XYZ}(x,y,z)
 \equiv \mu_{YZ|X}(y,z|x)\mu_{X}(x).
 \label{eq:markov-wiretap-capable}
\end{align}
This capacity formula is derived in \cite{W75b}
for a degraded broadcast channel,
extended in \cite{CK78} to the case where a channel between
$X$ and $Y$ is more capable than a channel between $X$ and $Z$.

\begin{figure}[t]
 \begin{center}
  \unitlength 0.4mm
  \begin{picture}(156,99)(0,0)
   \put(82,89){\makebox(0,0){Encoder}}
   \put(54,61){\makebox(0,0){$\cc$}}
   \put(59,61){\vector(1,0){10}}
   \put(24,43){\makebox(0,0){$\mm$}}
   \put(34,43){\vector(1,0){35}}
   \put(24,25){\makebox(0,0){$\ww$}}
   \put(34,25){\vector(1,0){35}}
   \put(69,18){\framebox(26,50){$g_{AB\hB}$}}
   \put(95,43){\vector(1,0){35}}
   \put(140,43){\makebox(0,0){$\xx$}}
   \put(44,5){\framebox(76,76){}}
  \end{picture}
  \\
  \begin{picture}(156,70)(0,0)
   \put(82,60){\makebox(0,0){Decoder}}
   \put(30,35){\makebox(0,0){$\cc$}}
   \put(35,35){\vector(1,0){10}}
   \put(0,17){\makebox(0,0){$\yy$}}
   \put(10,17){\vector(1,0){35}}
   \put(45,10){\framebox(18,32){$g_A$}}
   \put(63,26){\vector(1,0){10}}
   \put(83,26){\makebox(0,0){$\xx$}}
   \put(93,26){\vector(1,0){10}}
   \put(103,19){\framebox(30,14){$B$}}
   \put(133,26){\vector(1,0){20}}
   \put(167,26){\makebox(0,0){$\mm$}}
   \put(20,0){\framebox(124,52){}}
  \end{picture}
 \end{center}
\caption{Construction of Wiretap Channel Code}
\label{fig:wiretap-code}
\end{figure}
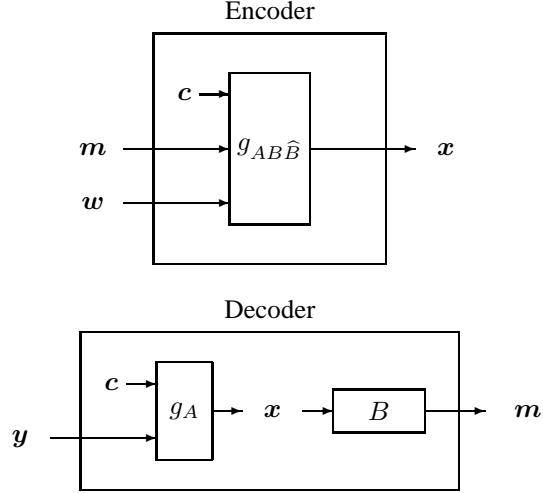

In the following, we assume that $\mu_{X}$ and $\mu_{YZ|X}$ are given,
where it is not necessary to assume that
a channel is degraded or a channel between $X$ and $Y$ is more capable
than that between $X$ and $Z$.
We fix functions
\begin{align*}
 A&:\X^n\to\X^{\lA}
 \\
 B&:\X^n\to\X^{\lB}
 \\
 \hB&:\X^n\to\X^{\lhB}
\end{align*}
and a vector $\cc\in\X^{\lA}$
available for an encoder, a decoder, and an eavesdropper, where
\begin{align*}
 \lA
 &\equiv \frac{n[H(X|Y)+\eA]}{\log|\X|}
 \\
 \lB
 &\equiv \frac{n[H(X|Z)-H(X|Y)]}{\log|\X|}
 \\
 &=\frac{n[I(X;Y)-I(X;Z)]}{\log|\X|}
 \\
 \lhB
 &\equiv \frac{n[I(X;Z)-\ehB]}{\log|\X|}.
\end{align*}

We construct a stochastic encoder and assume that
the encoder uses a random sequence $\ww\in\X^{\lhB}$,
which is generated uniformly at random and independently
of the channel and the message $\mm\in\X^{\lB}$.
We define the encoder and the decoder
\begin{align*}
 \Encoder&:\X^{\lB}\times\X^{\lhB}\to\X^n
 \\
 \Decoder&:\Y^n\to\X^{\lB}
\end{align*}
as
\begin{align*}
 \Encoder(\mm,\ww)
 &\equiv g_{AB\hB}(\cc,\mm,\ww)
 \\
 \Decoder(\yy)
 &\equiv Bg_A(\cc|\yy),
\end{align*}
where
$g_{AB\hB}(\cc,\mm,\ww)$ and $g_{A}(\cc|\yy)$ are
defined by (\ref{eq:md}) and (\ref{eq:ml}), respectively.
It is noted that $g_{AB\hB}$ is a deterministic map.

Let $M$ and $W$ be random variables corresponding to
$\mm$ and $\ww$, respectively,
where the probability distributions $p_M$ and $p_W$ are given by
\begin{align}
 p_M(\mm)
 &\equiv
 \begin{cases}
  \frac 1{|\im\B|}
  &\text{if}\ \mm\in\im\B
  \\
  0,
  &\text{if}\ \mm\notin\im\B
 \end{cases}
 \label{eq:pM}
 \\
 p_W(\ww)
 &\equiv
 \begin{cases}
  \frac 1{|\im\hcB|}
  &\text{if}\ \ww\in\im\hcB
  \\
  0,
  &\text{if}\ \ww\notin\im\hcB
 \end{cases}
 \label{eq:pW}
\end{align}
and the joint distribution $p_{MWYZ}$ of the messages,
and the channel outputs is
given by
\begin{equation*}
 p_{MWYZ}(\mm,\ww,\yy,\zz)
 \equiv
 \mu_{YZ|X}(\yy,\zz|\Encoder(\mm,\ww))p_M(\mm)p_W(\ww).
\end{equation*}

The rate of this code is given by
\begin{align*}
 \Rate
 &\equiv\frac{\log|\im\B|}n
 \\
 &=I(X;Y)-I(X;Z)-\frac{\log\frac{|\X|^{l_B}}{|\im\B|}}n
\end{align*}
which converges to $I(X;Y)-I(X;Z)$ as $n$ goes to infinity
by assuming the condition (\ref{eq:imA}) for an ensemble $(\B,p_B)$.
The decoding error probability $\Error_{Y|X}(A,B,\hB,\cc)$
is given by
\begin{equation}
 \Error_{Y|X}(A,B,\hB,\cc)
 \equiv
 \sum_{\mm,\ww,\yy}\mu_{Y|X}(\yy|\Encoder(\mm,\ww))p_M(\mm)p_W(\ww)
 \chi(\Decoder(\yy)\neq\mm).
 \label{eq:def-error-wiretap}
\end{equation}
The information leakage $\Security_{Z|X}(A,B,\hB,\cc)$ is given by
\begin{equation}
 \Security_{Z|X}(A,B,\hB,\cc)\equiv\frac {I(M;Z^n)}n.
  \label{eq:def-leakage-wiretap}
\end{equation}
It should be noted that
the vector $\cc$ is considered to be part of a deterministic map,
which is known by the eavesdropper.

We have the following theorem.
It should be noted that
alphabets $\X$ and $\Y$ is allowed to be non-binary,
and the channel is allowed to be asymmetric, non-degraded.
\begin{thm}
 \label{thm:wiretap}
 Let $\mu_{YZ|X}$ be the conditional probability distribution
 of a stationary memoryless channel.
 For given $\lA$ and $\lB$, assume that ensembles
 $(\bcA,\bp_A)$, $(\bcA\times\bcB,\bp_{AB})$,
 and $(\bcA\times\bcB\times\bchB,\bp_{AB\hB})$
 have a hash property.
 Then for any $\delta>0$ and all sufficiently large $n$,
 there are $\ehB>\eA>0$,
 functions (sparse matrices) $A\in\A$, $B\in\B$, $\hB\in\hcB$,
 and a vector $\cc\in\im\A$ such that
 \begin{gather}
  \Rate > I(X;Y)-I(X;Z)-\delta
  \label{eq:rate-wiretap}
  \\
  \Error_{Y|X}(A,B,\hB,\cc)<\delta
  \label{eq:error-wiretap}
  \\
  \Security_{Z|X}(A,B,\hB,\cc)<\delta.
  \label{eq:leakagebound-wiretap}
 \end{gather}
 By assuming that the channel between $X$ and $Y$ is more capable
 than that between $X$ and $Z$,
 $\mu_X$ attains the secrecy capacity defined by
 (\ref{eq:wiretap-capacity-capable}), and $\delta\to 0$,
 the rate of the proposed code is close to the secrecy capacity.

 For a general wiretap channel $\mu_{YZ|X}$,
 let $F:\hcX\to\X$ be a channel (non-deterministic map)  corresponding
 to a conditional probability distribution $\mu_{X|\hX}$ and assume that
 \[
  \mu_{\hX X}(\hx,x)\equiv\mu_{X|\hX}(x|\hx)\mu_{\hX}(\hx)
 \]
 achieves the maximum of the right hand side of
 (\ref{eq:wiretap-capacity}).
 By using a proposed code for the channel $\mu_{YZ|\hX}$
 defined as
 \[
  \mu_{YZ|\hX}(y,z|\hx)\equiv\sum_{x}\mu_{YZ|X}(y,z|x)\mu_{X|\hX}(x|\hx)
 \]
 with the input distribution $\mu_{\hX}$, we construct a code
 for the channel $\mu_{YZ|X}$ as
 \begin{align*}
  \Encoder(\mm,\ww)
  &\equiv \FF(g_{AB\hB}(\cc,\mm,\ww))
  \\
  \Decoder(\yy)
  &\equiv Bg_A(\cc|\yy),
 \end{align*}
 where $g_{AB\hB}(\cc,\mm,\ww)$
 outputs the channel input  $\hxx\equiv(\hx_1,\ldots,\hx_n)\in\hcX^n$ of
 outer channel  $\mu_{YZ|\hX}$,
 $\FF$ is defined as
 \[
  \FF(\hxx)\equiv(F(\hx_1),\ldots,F(\hx_n)),
 \]
 and $g_A(\cc|\yy)$ reproduces $\hxx$ with small error probability.
 Then the rate of this code is
 close to the secrecy capacity
 of the channel $\mu_{YZ|X}$ defined by (\ref{eq:wiretap-capacity}).
\end{thm}

\section{Universal Wiretap Channel Coding}

In this section we consider 
the fixed-rate universal wiretap channel coding
for any stationary memoryless channel $\mu_{YZ|X}$,
where an input distribution $\mu_X$ is given
and it is enough to know the upper bound of $H(X|Y)$ and 
the lower bound of $I(X;Z)$ before constructing the code.
It should be noted here that we have to know
the sizes of $\X$, $\Y$, and $\Z$ in advance.

For a given $\RA,\RB>0$, let $p_A$ and $p_B$ be ensembles of functions
\begin{align*}
 A&:\X^n\to\X^{\lA}
 \\
 B&:\X^n\to\X^{\lB}
 \\
 \hB&:\X^n\to\X^{\lhB}
\end{align*}
satisfying
\begin{align*}
 \RA&=\frac{\log|\im\A|}n
 \\
 \RB&=\frac{\log|\im\B|}n
 \\
 \RhB&=\frac{\log|\im\hcB|}n,
\end{align*}
respectively.
It should be noted that $\im\B$ represents the set of all messages,
$\RB$ represents the encoding rate of a confidential message.

We fix functions $A$, $B$, $\hB$, and a vector $\cc\in\X^{\lA}$
available for an encoder, a decoder, and an eavesdropper.
We construct a stochastic encoder and assume that
the encoder uses a random sequence $\ww\in\X^{\lhB}$,
which is generated uniformly at random and independently
of the channel and the message $\mm\in\X^{\lB}$.
We define the same encoder and decoder
as defined in the last section except to replace
and $g_A$ by $\tg_A$ defined by (\ref{eq:me}).

Let $M$ and $W$ be random variables corresponding to
$\mm$ and $\ww$, respectively,
where the probability distributions $p_M$ and $p_W$ are given by
(\ref{eq:pM}) and (\ref{eq:pW}), respectively.
The decoding error probability $\Error_{Y|X}(A,B,\hB,\cc)$
and the information leakage $\Security_{Z|X}(A,B,\hB,\cc)$
are given by (\ref{eq:def-error-wiretap}) 
and (\ref{eq:def-leakage-wiretap}), respectively.

We have the following theorem.
It should be noted that
alphabets $\X$ and $\Y$ is allowed to be non-binary,
and the channel is allowed to be asymmetric.
\begin{thm}
 \label{thm:universal}
 For $\RA$, $\RB$, and $\RhB$,
 Assume that ensembles
 $(\bcA,\bp_A)$, $(\bcA\times\bcB,\bp_{AB})$,
 and $(\bcA\times\bcB\times\bchB,p_{AB\hB})$
 have a hash property.
 Let $\mu_{X}$ be the distribution of the channel input
 satisfying
 \begin{equation}
  \RA+\RB+\RhB<H(X),
   \label{eq:rate-HX}
 \end{equation}
 where $\RB$ represents the encoding rate of a confidential message.
 Then for any $\delta>0$ and all sufficiently large $n$,
 there are functions (sparse matrices) $A\in\A$, $B\in\B$, $\hB\in\hcB$,
 and a vector $\cc\in\im\A$ such that
 \begin{gather}
  \Error_{Y|X}(A,B,\hB,\cc)<\delta
  \label{eq:error-uwiretap}
  \\
  \Security_{Z|X}(A,B,\hB,\cc)<\delta
  \label{eq:leakbound-uwiretap}
 \end{gather}
 for any stationary memoryless channel $\mu_{YZ|X}$ 
 satisfying
 \begin{gather}
  \RA>H(X|Y)
  \label{eq:rate-HXgY}
  \\
  \RhB\geq I(X;Z).
  \label{eq:rate-IXZ}
 \end{gather}
\end{thm}
\begin{rem}
 It should be noted that (\ref{eq:rate-HX}), (\ref{eq:rate-HXgY}),
 and (\ref{eq:rate-IXZ}) imply
 \begin{gather*}
  0<\RA < H(X|Z)
  \\
  0<\RB < I(X;Y)-I(X;Z).
 \end{gather*}
\end{rem}

\section{Secret Key Agreement from Correlated Source Outputs}
\label{sec:ska}
In this section we construct codes for secret key agreement
from the correlated source outputs $(X,Y,Z)$ introduced in \cite{Mau93}
(see Fig.\ \ref{fig:ska}),
where a sender, a receiver, and an eavesdropper
have access to $X$, $Y$, and $Z$, respectively.
The secret key capacity, which represents
the optimal key generation rate, is given in \cite{SKCADC} as
\begin{equation}
 \Capacity=\sup_{n,t,(C_1^t,\hX,\hY)} \frac 1n\lrB{I(\hX;\hY)-I(\hX;Z,C_1^t)},
  \label{eq:skc}
\end{equation}
where the supremum is taken over all $n$, $t$, and protocols $(C_1^t,\hX,\hY)$
satisfying Markov conditions
\begin{gather*}
 Y^nZ^n C_{i+1}^t\hX\hY\markov X^n C_1^{i-1}\markov C_i,
 \ \text{if $i$ is odd}
 \\
 X^nZ^n C_{i+1}^t\hX\hY\markov Y^nC_1^{i-1}\markov C_i,
 \ \text{if $i$ is even}
 \\
 Y^nZ^n\hY\markov X^nC_1^t\markov \hX
 \\
 X^nZ^n\hY\markov Y^nC_1^t\markov \hY
\end{gather*}
in which $C_1^t$ represents the communication between the sender and the
receiver via a public channel and finally the sender and the receiver
generate $\hX$ and $\hY$, respectively.
It should be noted that $\hX\neq\hY$ is allowed with high probability.
According to \cite{BBCM95}\cite{CM97},
there are three steps in a secret key agreement:
advantage distillation, information reconciliation, and 
privacy amplification.
This section deals with
the combination of information reconciliation
and privacy amplification studied
in \cite{AC93}\cite{CM97}\cite{ISIT04}\cite{CRYPTLDPC}.
In the following, we assume that a fixed joint distribution
$\mu_{XYZ}$ satisfies
\[
 I(X;Y)-I(X;Z)=H(X|Z)-H(X|Y)>0
\]
and do not deal with advantage distillation.
From (\ref{eq:skc}),
we can construct a protocol whose rate is close to the secret key capacity
by combining an advantage distillation protocol $(C_1^t,\hX,\hY)$
with the following one-way secret key agreement protocol,
where $I(\hX;\hY)-I(\hX;Z,C_1^t)$ is close to the secret key capacity.

In the following, we focus on the one way secret key agreement protocol.
When secret key agreement is allowed to be one-way from the sender
to the receiver,
the forward secret key capacity is given in \cite{AC93} by
\begin{equation}
 \Capacity=\max_{C,\hX} \lrB{I(\hX;Y|C)-I(\hX;Z|C)},
  \label{eq:owskc}
\end{equation}
where the maximum is taken over all random variables $C$ and $\hX$
that satisfy the Markov condition
\[
 \hX\longleftrightarrow C\longleftrightarrow X\longleftrightarrow YZ.
\]
Since
\[
 I(\hX;Y|C)-I(\hX;Z|C)=I(\hX;Y,C)-I(\hX;Z,C),
\]
then we can construct an optimal one-way secret key agreement protocol
by applying the following protocol
to the correlated source $(\hX,(Y,C),(Z,C))$,
which achieves the maximum on the right hand side of (\ref{eq:owskc}).

The following construction is based on \cite{CRYPTLDPC}.
We fix functions
\begin{align*}
 A&:\X^n\to\X^{\lA}
 \\
 B&:\X^n\to\X^{\lB}
\end{align*}
available for an encoder, a decoder, and an eavesdropper, where
\begin{align*}
 \lA
 &\equiv \frac{n[H(X|Y)+\eA]}{\log|\X|}
 \\
 \lB
 &\equiv \frac{n[H(X|Z)-H(X|Y)]}{\log|\X|}
 \\
 &=\frac{n[I(X;Y)-I(X;Z)]}{\log|\X|}.
\end{align*}

Then a secret key agreement protocol
is described below (see Fig.\ \ref{fig:owska}).
\paragraph*{\bf Encoding}
Let $\xx\in\X^n$ be a sender's random sequence.
The sender transmits $\cc$ to a legitimate receiver via a public
channel and generates a secret key by $\mm$,
where
$\cc$ and $\mm$ are defined as
\begin{align}
 \cc&\equiv A\xx
 \label{eq:ska-C}
 \\
 \mm&\equiv B\xx,
 \label{eq:ska-M}
\end{align}
respectively.

\paragraph*{\bf Decoding}
Let $\yy\in\Y^n$ be a receiver's random sequence,
and $\cc\equiv A\xx$ be a codeword received from the sender
via a public channel.
The receiver generates a secret key by $Bg_{A}(\cc|\yy)$,
where $g_A$  is defined by (\ref{eq:ml}).

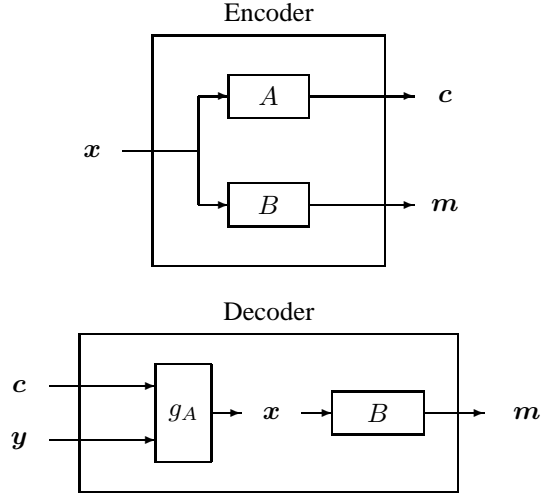
\begin{figure}[t]
 \begin{center}
  \unitlength 0.4mm
  \begin{picture}(156,99)(0,0)
   \put(82,89){\makebox(0,0){Encoder}}
   \put(24,43){\makebox(0,0){$\xx$}}
   \put(34,43){\line(1,0){25}}
   \put(59,25){\line(0,1){36}}
   \put(59,61){\vector(1,0){10}}
   \put(59,25){\vector(1,0){10}}
   \put(69,54){\framebox(26,14){$A$}}
   \put(69,18){\framebox(26,14){$B$}}
   \put(140,61){\makebox(0,0){$\cc$}}
   \put(140,25){\makebox(0,0){$\mm$}}
   \put(95,61){\vector(1,0){35}}
   \put(95,25){\vector(1,0){35}}
   \put(44,5){\framebox(76,76){}}
  \end{picture}
  \\
  \begin{picture}(156,70)(0,0)
   \put(82,60){\makebox(0,0){Decoder}}
   \put(0,35){\makebox(0,0){$\cc$}}
   \put(10,35){\vector(1,0){35}}
   \put(0,17){\makebox(0,0){$\yy$}}
   \put(10,17){\vector(1,0){35}}
   \put(45,10){\framebox(18,32){$g_A$}}
   \put(63,26){\vector(1,0){10}}
   \put(83,26){\makebox(0,0){$\xx$}}
   \put(93,26){\vector(1,0){10}}
   \put(103,19){\framebox(30,14){$B$}}
   \put(133,26){\vector(1,0){20}}
   \put(167,26){\makebox(0,0){$\mm$}}
   \put(20,0){\framebox(124,52){}}
  \end{picture}
 \end{center}
 \caption{Construction of One-way Secret Key Agreement Protocol}
 \label{fig:owska}
\end{figure}

Let $C$ and $M$ be random variables corresponding to $\cc$ and $\mm$
defined by (\ref{eq:ska-C}) and (\ref{eq:ska-M}), respectively.
The key generation rate is given by
\begin{align}
 \Rate&\equiv\frac{H(M)}n.
 \label{eq:def-rate-ska}
\end{align}
The error probability
of the secret key agreement is given by
\begin{equation}
 \Error_{XY}(A,B)\equiv
  \mu_{XY}\lrsb{\lrb{
  (\xx,\yy): Bg_A(A\xx,\yy)\neq B\xx
  }}.
  \label{eq:def-error-ska}
\end{equation}
The information leakage is given by
\begin{equation}
 \Security_{XYZ}(A,B)\equiv\frac {I(M;Z^n,C)}n.
 \label{eq:def-leakage-ska}
\end{equation}

We have the following theorem.
\begin{thm}
\label{thm:ska}
 For given $\lA$ and $\lB$,
 assume that ensembles
 $(\bcA,\bp_A)$ and $(\bcA\times\bcB,\bp_{AB})$
 have a hash property.
 For all $\delta>0$ and sufficiently large $n$,
 there are $\eA>0$ and functions (sparse matrices) $A\in\B$ and $B\in\B$
 such that the above secret key agreement protocol satisfies
 \begin{gather}
  \Rate > I(X;Y)-I(X;Z)-\delta
  \label{eq:rate-ska}
  \\
  \Error_{XY}(A,B)< \delta
  \label{eq:error-ska}
  \\
  \Security_{XYZ}(A,B)< \delta.
  \label{eq:leakage-ska}
 \end{gather}

 By assuming that random variables $C$ and $\hX$ attain
 the forward secret key capacity given by (\ref{eq:owskc})
 and the sender sends message $C$ via public channel
 before the protocol,
 the rate of the proposed secret key agreement protocol
 for correlated sources $(\hX,(Y,C),(Z,C))$
 is closed to the forward secret key capacity.
\end{thm}

\section{Universal Secret Key Agreement from Correlated Source Outputs}
\label{sec:ska-universal}
In this section, we
construct a fixed-rate universal secret key agreement scheme
for any stationary memoryless sources $(X,Y,Z)$,
where
it is enough to know the upper bound of $H(X|Y)$ and 
the lower bound of $H(X|Z)$
before constructing the code.
It should be noted here that we have to know
the sizes of $\X$, $\Y$, and $\Z$ in advance.

For a given $\RA, \RB>0$, let $p_A$ and $p_B$ be
ensembles of functions
\begin{align*}
 A&:\X^n\to\X^{\lA}
 \\
 B&:\X^n\to\X^{\lB},
\end{align*}
where
\begin{align*}
 \lA&\equiv \frac{n\RA}{\log|\X|}
 \\
 \lB&\equiv \frac{n\RB}{\log|\X|}.
\end{align*}

We use the same secret key agreement protocol as that
described in the last section except that we
replace $g_A$ by $\tg_A$ defined by (\ref{eq:me}).

The key generation rate $\Rate$,
the error probability $\Error_{XY}(A,B)$,
and the information leakage $\Security(A,B)$
are defined by (\ref{eq:def-rate-ska}), (\ref{eq:def-error-ska}),
and (\ref{eq:def-leakage-ska}), respectively.

We have the following theorem.
\begin{thm}
 \label{thm:uska}
 For given $\RA$ and $\RB$, assume that ensembles
 $(\bcA,\bp_A)$ and $(\bcA\times\bcB,\bp_{AB})$
 have a hash property.
 For all $\delta>0$ and sufficiently large $n$,
 there are functions (sparse matrices) $A\in\A$ and $B\in\B$ such that
 the above secret key agreement scheme satisfies
 \begin{gather}
  \Rate> \RB-\delta
  \label{eq:rate-uska}
  \\
  \Error_{XY}(A,B)< \delta
  \label{eq:error-uska}
  \\
  \Security_{XYZ}(A,B)< \delta.
  \label{eq:leakagebound-uska}
 \end{gather}
 for any stationary memoryless source $(X,Y,Z)$ satisfying
 \begin{align}
  \RA&>H(X|Y)
  \label{eq:RA-uska}
  \\
  \RA+\RB&\leq H(X|Z).
  \label{eq:RAB-uska}
 \end{align}
\end{thm}
\begin{rem}
 It should be noted that
 (\ref{eq:RA-uska}) and (\ref{eq:RAB-uska}) imply
 \begin{gather*}
  0<\RA<H(X|Z)
  \\
  0<\RB<I(X;Y)-I(X;Z).
 \end{gather*}
\end{rem}

\section{Applying Linear Programming Method to Minimum-divergence
Encoder, Maximum-likelihood Decoder, and Minimum-entropy Decoder}
\label{sec:binary}

In this section, we apply the linear programming method introduced by
\cite{FWK05}\cite{CME05} by assuming that $\X=\Y=\GF(2)$ and
$A$, $B$, and $\hB$ are linear functions (sparse matrices).
It should be noted again that
this method may not find integral solutions
and we do not discuss the performance of the linear programming methods
in this paper.

First, we construct the minimum-divergence encoder $g_{AB\hB}$
defined by (\ref{eq:md}).
The following construction is presented in \cite{HASH-UNIV}.
We use the fact that the analysis of error probability in the proof of
theorems is not changed if
we replace the minimum-divergence encoder $g_{AB\hB}$ by
\begin{equation*}
 g'_{AB\hB}(\cc,\mm,\ww)
  \equiv
  \begin{cases}
   \xx'
   &\text{if $\xx'\in\C_{AB\hB}(\cc,\mm,\ww)\cap\T_U$ exists,}
   \\
   \text{`error'}
   &\text{otherwise},
  \end{cases}
  \label{eq:cwe}
\end{equation*}
where $U$ is defined by (\ref{eq:U}) which appears in Appendix \ref{sec:lemma-uwiretap}.
Let
\begin{equation}
 t\equiv\arg\min_{t'\in\{0,1,\ldots,n\}}D(\nu_{t'}\|\mu_X),
  \label{eq:t}
\end{equation}
where $(\nu_t(0),\nu_t(1))\equiv(1-t/n,t/n)$.
Then the function $g'_{AB\hB}$ is realized by finding
$\xx'$ that satisfies
$A\xx'=\cc$, $B\xx'=\mm$, $\hB\xx'=\ww$,
and $\sum_{i=1}^n x'_i = t$
and declaring the encoding error if there is no such $\xx'$
that satisfies $A\xx'=\cc$, $B\xx'=\mm$, $\hB\xx'=\ww$,
and $\sum_{i=1}^n x'_i = t$,
where we consider $\xx'$ as a real-valued vector
in the third condition.
It should be noted that
it is realized by the linear programming method
because the conditions $A\xx'=\cc$, $B\xx'=\mm$, $\hB\xx'=\ww$
can be represented by
linear inequalities by using the technique of \cite{FWK05}.

Next, we construct the maximum-likelihood decoder $g_{A}$
defined by (\ref{eq:ml}).
The following construction is equivalent to \cite{FWK05}.
The function $g_A$ is realized by
\begin{gather*}
 g_{A}(\cc|\yy)
 \equiv
 \begin{cases}
  \displaystyle
  \arg\min_{\xx': A\xx'=\cc}\sum_{i=1}^n x'_i,
  &\text{if}\ 0\leq\mu_{X|Y}(1|0),\mu_{X|Y}(1|1)\leq 1/2
  \\
  \displaystyle
  \arg\min_{\xx': A\xx'=\cc}\sum_{i=1}^n [-1]^{y_i}x'_i,
  &\text{if}\ 0\leq\mu_{X|Y}(1|0),\mu_{X|Y}(0|1)\leq 1/2
  \\
  \displaystyle
  \arg\max_{\xx': A\xx'=\cc}\sum_{i=1}^n [-1]^{y_i}x'_i,
  &\text{if}\ 0\leq\mu_{X|Y}(0|0),\mu_{X|Y}(1|1)\leq 1/2
  \\
  \displaystyle
  \arg\max_{\xx': A\xx'=\cc}\sum_{i=1}^n x'_i,
  &\text{if}\ 0\leq\mu_{X|Y}(0|0),\mu_{X|Y}(0|1)\leq 1/2,
 \end{cases}
\end{gather*}
where $\xx'$ and $\yy$ are considered as real-valued vectors
in $\sum_{i=1}^nx'_i$ and $\sum_{i=1}^n[-1]^{y_i}x_i'$.
The above minimizations and maximizations are the linear programming
problems because the condition $A\xx'=\cc$ can be represented
by linear inequalities by using the technique of \cite{FWK05}.

Finally, we construct the minimum-entropy decoder $\tg_{A}$
defined by (\ref{eq:me}).
The following construction is presented in \cite{HASH-UNIV},
which is based on the idea presented in \cite{CME05}.
The function $\tg_A$ can be realized as
\begin{align}
 \xx_{t,\min}
 &\equiv\arg\min_{\substack{
 \xx':
 \\
 A\xx'=\cc
 \\
 \sum_{i=1}^n x'_i=t
 }}
 \sum_{i=1}^ny_ix'_i
 \label{eq:xmin}
 \\
 \xx_{t,\max}
 &\equiv\arg\max_{\substack{
 \xx':
 \\
 A\xx'=\cc
 \\
 \sum_{i=1}^n x'_i=t
 }}
 \sum_{i=1}^ny_ix'_i
 \label{eq:xmax}
 \\
 \tg_{A}(\cc|\yy)
 &\equiv\arg\min_{\xx'\in\cup_{t=0}^n\{\xx_{t,\min},\xx_{t,\max}\}}
 H(\xx'|\yy),
 \label{eq:gacy}
\end{align}
where $\xx'$ and $\yy$ are considered as real-valued vectors
in $\sum_{i=1}^nx'_i$ and $\sum_{i=1}^ny_ix_i'$.
The derivation of (\ref{eq:gacy}) is presented
in \cite[Appendix A]{HASH-UNIV}.
We can use the linear programming method to obtain
$\xx_{t,\min}$ and
$\xx_{t,\max}$ because the constraint $A\xx'=\cc$
can be represented by linear inequalities
by using the technique introduced in \cite{FWK05}.
It should be noted that
$g_A$ can be replaced by
\begin{align}
 g'_{A}(\cc|\yy)
 &\equiv\arg\min_{\xx'\in\C_A(\cc)\cap\T_U}H(\xx'|\yy)
 \notag
 \\
 &=\arg\min_{\xx'\in\{\xx_{t,\min},\xx_{t,\max}\}}
 H(\xx'|\yy)
\end{align}
by assuming that $U$ 
defined by (\ref{eq:U}) or $t$ defined by (\ref{eq:t})
is shared by the encoder and the decoder,
where $\xx_{t,\min}$ and $\xx_{t,\max}$ are defined by
(\ref{eq:xmin}) and (\ref{eq:xmax}), respectively.

\section{Proof of Theorems}

\subsection{Proof of Theorem~\ref{thm:wiretap}}
We use the following lemma which is proved in Appendix.
\begin{lem}
 \label{lem:error-wiretap}
 Let $g_{AB}(\cc,\mm|\zz)$ be defined as
 \[
 g_{AB}(\cc,\mm|\zz)
 \equiv
 \arg\max_{\xx'\in\C_{AB}(\cc,\mm)}\mu_{X|Z}(\xx'|\zz).
 \]

 Then, for all $\delta'>0$, all sufficiently small $\gamma>0$,
 and all sufficiently large $n$,
 there are functions (sparse matrices)
 $A\in\A$, $B\in\B$, $\hB\in\hcB$, and a vector $\cc\in\im\A$ such that
 \begin{equation}
  p_{MWYZ}\lrsb{\lrb{
   (\mm,\ww,\yy,\zz):
   \begin{aligned}
    &g_{AB\hB}(\cc,\mm,\ww)\notin\T_{X,\gamma}
    \\
    &\text{or}\ \zz\notin\T_{Z|X,\gamma}(g_{AB\hB}(\cc,\mm,\ww))
    \\
    &\text{or}\ g_{AB\hB}(\cc,\mm,\ww)\neq g_{AB}(\cc,\mm|\zz)
    \\
    &\text{or}\ g_{AB\hB}(\cc,\mm,\ww)\neq g_{A}(\cc|\yy)
   \end{aligned}
   }}
   \leq
   \delta'.
   \label{eq:error-wiretap-lemma}
 \end{equation}
\end{lem}

Now we prove Theorem \ref{thm:wiretap}.
The equality (\ref{eq:rate-wiretap}) has already been shown.
Since $g_{AB\hB}(\cc,\mm,\ww)=g_A(\cc|\yy)$ implies
\begin{align*}
 \varphi^{-1}(\yy)
 &=
 Bg_A(\cc|\yy)
 \\
 &=
 Bg_{AB\hB}(\cc,\mm,\ww)
 \\
 &=
 \mm
\end{align*}
for all $\cc$ and $\ww$,
the inequality (\ref{eq:error-wiretap}) comes immediately from
Lemma~\ref{lem:error-wiretap} by letting $\delta'<\delta$.

In the following we prove (\ref{eq:leakagebound-wiretap}).
From Lemma~\ref{lem:error-wiretap} and Fano's inequality,
we have
\begin{align}
 H(g_{AB\hB}(\cc,M,W)|Z^n,M)
 &\leq h(\delta')+n\delta'\log|\X|
 \label{eq:fano}
\end{align}
for all $\delta>0$ and all sufficiently large $n$,
where $h$ is defined by (\ref{eq:def-h}).

Let $\txx\equiv g_{AB\hB}(\cc,\mm,\ww)$, $\tX^n\equiv g_{AB\hB}(\cc,M,W)$
and $\tcX$ be defined as
\[
 \tcX\equiv\lrb{
 g_{AB\hB}(\cc,\mm,\ww):
 \mm\in\X^{\lB}, \ww\in\X^{\lhB}
 }.
\]
Then the probability distribution $P_{\tX Z}$ is given by
\begin{align}
 P_{\tX Z}(\txx,\zz)
 &=
 \sum_{\substack{
 \mm,\ww:\\
 \txx=g_{AB\hB}(\cc,\mm,\ww)
 }}
 \mu_{Z|X}(\zz|\txx)P_M(\mm)P_W(\ww),
 \notag
 \\
 &=
 \begin{cases}
  \frac{\mu_{Z|X}(\zz|\txx)}{|\im\B||\im\hcB|},
  &\text{if}\ \txx\in\tcX, \mm\in\im\B, \ww\in\im\hcB
  \\
  0,
  &\text{otherwise}
 \end{cases}
 \label{eq:ptXZ}
\end{align}
where the summation equals zero when $\txx\notin\tcX$
and the second equality comes from the fact that
if $\txx\in\tcX$ then
there is a unique pair $(\mm,\ww)$ such that
$\txx=g_{AB\hB}(\cc,\mm,\ww)$.
From Lemma \ref{lem:typical-aep}, we have
\begin{align}
 \mu_{Z|X}(\zz|\xx)
 &\leq
 2^{-n[H(Z|X)-\zeta_{\Z|\X}(\gamma|\gamma)]}
 \label{eq:muZgX}
\end{align}
for $\xx\in\T_{X,\gamma}$ and $\zz\in\T_{Z|X,\gamma}(\xx)$.
Then the joint entropy $H(\tX^n,Z^n)$ is given by
\begin{align}
 H(\tX^n,Z^n)
 &\geq
 \sum_{\txx\in\T_{X,\gamma}}\sum_{\zz\in\T_{Z|X,\gamma}(\txx)}
 P_{\tX Z}(\txx,\zz)
 \log\frac 1{P_{\tX,Z}(\txx,\zz)}
 \notag
 \\
 &\geq
 \sum_{\txx\in\T_{X,\gamma}}\sum_{\zz\in\T_{Z|X,\gamma}(\txx)}
 P_{\tX Z}(\txx,\zz)
 \lrB{n[H(Z|X)-\zeta_{\Z|\X}(\gamma|\gamma)]+\log|\im\B||\im\hcB|}
 \notag
 \\
 &\geq
 n[1-\delta']
 \lrB{H(Z|X)+\frac1n\log|\im\B||\im\hcB|-\zeta_{\Z|\X}(\gamma|\gamma)}
 \notag
 \\
 &\geq
 n[H(Z|X)+I(X;Y)]-\log\frac{|\X|^{\lB+\lhB}}{|\im\B||\im\hcB|}
 -n\lrB{
 \delta'\log|\X||\Z|+\zeta_{\Z|\X}(\gamma|\gamma)+\ehB
 }
 \label{eq:HtXZ}
\end{align}
for sufficiently large $n$,
where the second inequality comes from (\ref{eq:ptXZ}) and (\ref{eq:muZgX}),
and the third inequality comes from Lemma~\ref{lem:error-wiretap}.
Then we have
\begin{align}
 I(M;Z^n)
 &=
 H(M)+H(Z^n)-H(Z^n,M)
 \notag
 \\
 &=
 H(M)+H(Z^n)
 -H(Z^n,M,g_{AB\hB}(\cc,M,W))
 +H(g_{AB\hB}(\cc,M,W)|Z^n,M)
 \notag
 \\
 &=
 H(M)+H(Z^n)
 -H(Z^n,g_{AB\hB}(\cc,M,W))
 +H(g_{AB\hB}(\cc,M,W)|Z^n,M)
 \notag
 \\
 &\leq
 H(M)+H(Z^n)
 -H(Z^n,g_{AB\hB}(\cc,M,W))
 +h(\delta')+n\delta'\log|\X|
 \notag
 \\
 &\leq
 n[I(X;Y)-I(X;Z)]+H(Z^n)
 -n[H(Z|X)+I(X;Y)]+\log\frac{|\X|^{\lB+\lhB}}{|\im\B||\im\hcB|}
 \notag
 \\*
 &\quad
 +n\lrB{\delta'\log|\X||\Z|+\zeta_{\Z|\X}(\gamma|\gamma)+\ehB}
 +h(\delta')+n\delta'\log|\X|
 \notag
 \\
 &<
 n\delta
\end{align}
for sufficiently large $n$,
where the third equality comes from the fact that $Bg(\cc,M,W)=M$,
the first inequality comes from (\ref{eq:fano}),
the second inequality comes from (\ref{eq:HtXZ}),
and we choose suitable $\ehB,\gamma,\delta'>0$
to satisfy the last inequality.
From this inequality we have (\ref{eq:leakagebound-wiretap}).

\subsection{Proof of Theorem~\ref{thm:universal}}

We use the following lemmas, which are proved in Appendix.
\begin{lem}
 \label{lem:R}
 If  $I(X;Z)\leq R$,
 then for all $\e>0$ there is a random variable  $\tZ$ taking values in 
 $\tcZ\equiv\X\times\Z$ and a function $f$ such that
 \begin{gather*}
  I(X;\tZ)=R+\e
  \\
  Z=f(\tZ).
 \end{gather*}
\end{lem}

\begin{lem}
 \label{lem:error-uwiretap}
 Let $\tg_{AB}(\cc,\mm|\tzz)$ be defined as
 \[
 \tg_{AB}(\cc,\mm|\tzz)
 \equiv
 \arg\min_{\xx'\in\C_{AB}(\cc,\mm)}H(\xx'|\tzz).
 \]
 Then,
 for all $\delta'>0$, all sufficiently small $\gamma>0$,
 and sufficiently large $n$,
 there are functions (sparse matrices)
 $A\in\A$, $B\in\B$, $\hB\in\hcB$,
 and a vector $\cc\in\im\A$ such that
 \begin{align}
  p_{MWY\tZ}\lrsb{\lrb{
  (\mm,\ww,\yy,\tzz):
  \begin{aligned}
   &\tg_{AB\hB}(\cc,\mm,\ww)\notin\T_{X,\gamma}
   \\
   &\text{or}\ \tzz\notin\T_{\tZ|X,\gamma}(\tg_{AB\hB}(\cc,\mm,\ww))
   \\
   &\text{or}\ \tg_{AB\hB}(\cc,\mm,\ww)\neq \tg_{A}(\cc|\yy)
   \\
   &\text{or}\ \tg_{AB\hB}(\cc,\mm,\ww)\neq \tg_{AB}(\cc,\mm|\tzz)
  \end{aligned}
  }}
  \leq
  \delta'
  \label{eq:error-uwiretap-lemma}
 \end{align}
 for any $\mu_{Y\tZ|X}$ satisfying
 \begin{gather}
  \RA+\RB+\RhB<H(X)
  \label{eq:RABhB-uwiretap-lemma}
  \\
  \RA>H(X|Y)
  \label{eq:RA-uwiretap-lemma}
  \\
  \RA+\RB>H(X|\tZ).
  \label{eq:RAB-uwiretap-lemma}
 \end{gather}
\end{lem}

Now we prove Theorem \ref{thm:universal}.
The inequality (\ref{eq:error-uwiretap}) is shown similarly to the
proof of (\ref{eq:error-wiretap}).

In the following we prove (\ref{eq:leakbound-uwiretap}).
From Lemma~\ref{lem:R},
there is $\tZ\in\tcZ$ such that
\begin{equation}
 I(X;\tZ)=\RhB+\e,
 \label{eq:uwiretap-RhB}
\end{equation}
where $\e>0$ is specified later.
From Lemma~\ref{lem:error-uwiretap} and Fano's inequality,
we have
\begin{align}
 H(\tg_{AB\hB}(\cc,M,W)|\tZ^n,M)
 &\leq h(\delta')+n\delta'\log|X|
 \label{eq:fano-universal}
\end{align}
for all $\delta'>0$ and sufficiently large $n$,
where $h$ is defined by (\ref{eq:def-h}).
Similarly to the proof of (\ref{eq:HtXZ}),
we have
\begin{align}
 H(\tZ^n,\tg_{AB\hB}(\cc,M,W))
 &\geq
 n[1-\delta']\lrB{
 H(\tZ|X)+\frac 1n\log|\im\B||\im\hcB|-\zeta_{\tcZ|\X}(\gamma|\gamma)
 }
 \notag
 \\
 &\geq
 nH(\tZ|X)+\log|\im\B||\im\hcB|
 -n\lrB{
 \delta'\log|\X||\tcZ|+\zeta_{\tcZ|\X}(\gamma|\gamma)
 }
 \notag
 \\
 &\geq
 n[H(\tZ|X)+\RB+\RhB]
 -n\lrB{
 \delta'\log|\X||\tcZ|+\zeta_{\tcZ|\X}(\gamma|\gamma)
 },
 \label{eq:HtXZ-universal}
\end{align}
where the second inequality comes from the fact that
$\RB+\RhB<H(X)\leq\log|\X|$.
Then we have
\begin{align}
 I(M;Z^n)
 &=
 I(M;f(\tZ_1),\ldots,f(\tZ_n))
 \notag
 \\
 &\leq
 I(M;\tZ^n)
 \notag
 \\
 &=
 H(M)+H(\tZ^n)
 -H(\tZ^n,M,\tg_{AB\hB}(\cc,M,W))
 +H(\tg_{AB\hB}(\cc,M,W)|\tZ^n,M)
 \notag
 \\
 &\leq
 H(M)+H(\tZ^n)
 -H(\tZ^n,\tg_{AB\hB}(\cc,M,W))
 +h(\delta')+n\delta'\log|\X|
 \notag
 \\
 &\leq
 n\RB+H(\tZ^n)
 -n[H(\tZ|X)+\RB+\RhB]
 +
 n\lrB{
 \delta'\log|\X||\tcZ|+\zeta_{\Z|\X}(\gamma|\gamma)
 }
 +h(\delta')+n\delta'\log|\X|
 \notag
 \\
 &\leq
 n[I(X;\tZ)-\RhB]
 +
 n\lrB{
 \delta'\log|\X||\tcZ|+\zeta_{\Z|\X}(\gamma|\gamma)
 }
 +h(\delta')+n\delta'\log|\X|
 \notag
 \\
 &\leq
 n\lrB{
 \e+\delta'\log|\X|^2|\tcZ|+\zeta_{\Z|\X}(\gamma|\gamma)
 }
 +h(\delta')
 \notag
 \\
 &<
 n\delta
\end{align}
where
the second inequality comes from (\ref{eq:fano-universal})
and $M=Bg_{AB\hB}(\cc,M,W)$,
the third inequality comes from (\ref{eq:HtXZ-universal}),
the fifth inequality comes from (\ref{eq:uwiretap-RhB}),
and we choose a suitable $\gamma>0$, a suitable $\e>0$, and a suitable
$\delta'>0$ to satisfy the last inequality.
From this inequality, we have (\ref{eq:leakbound-uwiretap}).

\subsection{Proof of Theorem~\ref{thm:ska}}
We use the following lemma which is proved in Appendix.
\begin{lem}
 \label{lem:error-ska}
 Let $g_{AB}(\cc,\mm|\zz)$ be defined as
 \[
 g_{AB}(\cc,\mm|\zz)
 \equiv
 \arg\max_{\xx'\in\C_{AB}(\cc,\mm)}\mu_{X|Z}(\xx'|\zz).
 \]
 Then,
 for any $\delta'>0$, and all sufficiently large $n$,
 there are functions (sparse matrices) $A\in\A$ and $B\in\B$ such that
 \begin{align}
  p_{XYZ}\lrsb{\lrb{
  (\xx,\yy,\zz):
  \begin{aligned}
   &g_{A}(A\xx|\yy)\neq \xx
   \\
   &\text{or}\ g_{AB}(A\xx,B\xx|\zz)\neq \xx
  \end{aligned}
 }}
  \leq
  \delta'.
  \label{eq:error-ska-lemma}
 \end{align}
\end{lem}

Now we prove Theorem \ref{thm:ska}.

First, we prove (\ref{eq:error-ska}).
Since $g_{A}(A\xx|\yy)=\xx$ implies $Bg_A(A\xx|\yy)=B\xx$,
then the inequality (\ref{eq:error-ska}) comes immediately from
Lemma~\ref{lem:error-ska} by letting $\delta'<\delta$.

Next, we prove (\ref{eq:leakage-ska}).
From Lemma~\ref{lem:error-ska} and Fano's inequality, we have
\begin{align*}
 H(X^n|Z^n,C,M)
 &\leq h(\delta')+n\delta'\log|X|
\end{align*}
for all $\delta>0$ and all sufficiently large $n$,
where $h$ is defined by (\ref{eq:def-h}).
This implies that
\begin{align}
 H(Z^n,C,M)
 &\geq H(X^n,Z^n,C,M)
 -h(\delta')-n\delta'\log|X|
 \notag
 \\
 &= H(X^n,Z^n)
 -h(\delta')-n\delta'\log|X|
 \label{eq:fano-ska}
\end{align}
for all $\delta>0$ and all sufficiently large $n$,
where the equality comes from the definitions (\ref{eq:ska-C}) and
(\ref{eq:ska-M}) of $C$ and $M$.
Then we have
\begin{align}
 I(M;Z^n,C)
 &=
 H(Z^n,C)+H(M)-H(Z^n,C,M)
 \notag
 \\
 &\leq
 H(Z^n)+H(C)+H(M)-H(Z^n,C,M)
 \notag
 \\
 &\leq
 H(Z^n)+H(C)+H(M)
 -H(X^n,Z^n)
 +h(\delta')+n\delta'\log|X|
 \notag
 \\
 &\leq
 H(Z^n)+n[H(X|Y)+\eA]+n[H(X|Z)-H(X|Y)]
 -H(X^n,Z^n)
 +h(\delta')+n\delta'\log|X|
 \notag
 \\
 &=
 n\eA+h(\delta')+n\delta'\log|X|
 \notag
 \\
 &<
 n\delta,
\end{align}
where
the second inequality comes from (\ref{eq:fano-ska}), 
the third inequality comes from the definitions (\ref{eq:ska-C})
and (\ref{eq:ska-M}) of $C$ and $M$,
and we choose a suitable $\eA>0$ and a suitable $\delta'>0$ to satisfy
the last inequality.
From this inequality we have (\ref{eq:leakage-ska}).

Finally, we prove (\ref{eq:rate-ska}).
We have
\begin{align}
 H(M)
 &=
 H(M)+H(Z^n,C)-H(Z^n,C)
 \notag
 \\
 &\geq
 H(Z^n,C,M)-H(Z^n)-H(C)
 \notag
 \\
 &\geq
 H(X^n,Z^n)-h(\delta')-n\delta'\log|X|
 -H(Z^n)-n[H(X|Y)+\eA]
 \notag
 \\
 &=
 n[I(X;Z)-I(X;Y)]
 -n\eA-h(\delta')-n\delta'\log|X|
 \notag
 \\
 &\geq
 n[I(X;Z)-I(X;Y)]-n\delta,
\end{align}
where the second inequality comes from (\ref{eq:fano-ska}),
and we choose a suitable $\eA>0$ and a suitable $\delta'>0$ to satisfy
the last inequality.
From this inequality we have (\ref{eq:rate-ska}).

\subsection{Proof of Theorem \ref{thm:uska}}
\label{sec:proof-uska}

We use the following lemmas which are proved in Appendix.
\begin{lem}
 \label{lem:R-uska}
 If  $H(X|Z)\geq R$,
 then for all $\e>0$ there is a random variable  $\tZ$ taking values in 
 $\tcZ\equiv\X\times\Z$ and a function $f$ such that
 \begin{gather*}
  H(X|\tZ)=R-\e
  \\
  Z=f(\tZ).
 \end{gather*}
 \hfill\QED
\end{lem}

\begin{lem}
 \label{lem:error-uska}
 Let $\tg_{AB}(\cc,\mm|\tzz)$ be defined as
 \[
 \tg_{AB}(\cc,\mm|\tzz)
 \equiv
 \arg\min_{\xx'\in\C_{AB}(\cc,\mm)}H(\xx'|\tzz).
 \]
 Then,
 for any $\delta'>0$, and all sufficiently large $n$,
 there are functions (sparse matrices) $A\in\A$ and $B\in\B$ such that
 \begin{align}
  p_{XY\tZ}\lrsb{\lrb{
  (\xx,\yy,\tzz):
  \begin{aligned}
   &g_{A}(A\xx|\yy)\neq \xx
   \\
   &\text{or}\ g_{AB}(A\xx,B\xx|\tzz)\neq \xx
  \end{aligned}
  }}
  \leq
  \delta'
  \label{eq:error-uska-lemma}
 \end{align}
 for any $\mu_{XY\tZ}$ satisfying
 \begin{gather}
  \RA>H(X|Y)
  \label{eq:RA-uska-lemma}
  \\
  \RA+\RB>H(X|\tZ).
  \label{eq:RAB-uska-lemma}
 \end{gather}
 \hfill\QED
\end{lem}

Now we prove Theorem \ref{thm:uska}.
In the following, we prove
(\ref{eq:leakagebound-uska}) and (\ref{eq:rate-uska}).
The proof of (\ref{eq:error-uska}) is similar to that of 
(\ref{eq:error-uska-lemma}).

First, from Lemma~\ref{lem:R-uska} and (\ref{eq:RAB-uska-lemma}),
there is $\tZ\in\tcZ$ such that
\begin{equation}
 H(X|\tZ)=\RA+\RB-\e,
 \label{eq:uska-RAB}
\end{equation}
where $\e>0$ is specified later.
From Lemma~\ref{lem:error-uska} and Fano's inequality,
we have
\begin{align*}
 H(X^n|\tZ^n,C,M)
 &\leq h(\delta')+n\delta'\log|X|
\end{align*}
for all $\delta>0$ and all sufficiently large $n$,
where $h$ is defined by (\ref{eq:def-h}).
This implies that
\begin{align}
 H(\tZ^n,C,M)
 &\geq H(X^n,\tZ^n)-h(\delta')-n\delta'\log|X|.
 \label{eq:fano-uska}
\end{align}

Next, we prove (\ref{eq:leakagebound-uska}).
We have
\begin{align}
 I(M;Z^n,C)
 &=
 I(M;f(\tZ_1),\ldots,f(\tZ_n),C)
 \notag
 \\
 &\leq
 I(M;\tZ^n,C)
 \notag
 \\
 &\leq
 H(\tZ^n)+H(C)+H(M)-H(\tZ^n,C,M)
 \notag
 \\
 &\leq
 H(\tZ^n)+H(C)+H(M)
 -H(X^n,\tZ^n)
 +h(\delta')+n\delta'\log|X|
 \notag
 \\
 &\leq
 H(\tZ^n)+n\RA+n\RB
 -H(X^n,\tZ^n)
 +h(\delta')+n\delta'\log|X|
 \notag
 \\
 &=
 n\e+h(\delta')+n\delta'\log|X|
 \notag
 \\
 &<
 n\delta,
\end{align}
where
the third inequality comes from (\ref{eq:fano-uska}),
the fourth inequality comes from the definitions
(\ref{eq:ska-C}) and (\ref{eq:ska-M}) of $C$ and $M$,
the second equality comes from (\ref{eq:uska-RAB})
and we choose a suitable $\e>0$ and a suitable $\delta'>0$ to satisfy
the last inequality.
From this inequality, we have (\ref{eq:leakagebound-uska}).

Finally, we prove (\ref{eq:rate-uska}).
We have
\begin{align}
 H(M)
 &\geq
 H(\tZ^n,C,M)-H(\tZ^n)-H(C)
 \notag
 \\
 &\geq
 H(X^n,\tZ^n)-h(\delta')-n\delta'\log|X|
 -H(\tZ^n)-n\RA
 \notag
 \\
 &=
 n\RB-n\e-h(\delta')-n\delta'\log|X|,
 \notag
 \\
 &\geq
 n\RB-n\delta,
\end{align}
where the second inequality comes from (\ref{eq:fano-uska}),
the equality comes from (\ref{eq:uska-RAB}),
and we choose a suitable $\e>0$ and a suitable $\delta'>0$ to satisfy
the last inequality.
From this inequality, we have (\ref{eq:rate-uska}).

\section{Conclusion}
The constructions of codes for the wiretap channel
and secret key agreement from correlated source outputs
were presented.
The optimality, reliability, and security of the codes were proved
and the universal reliability and security were also proved.
The proof of the theorems
is based on the notion of a hash property for an ensemble of functions.
Since an ensemble of sparse matrices has a hash property,
we can construct codes by using sparse matrices
and practical encoding and decoding methods
are expected to be effective.
We believe that our construction can be 
applied to a quantum channel to realize a quantum cryptography.
However, it should be noted that
the security criteria should be revised to the quantum version.

\appendix
\subsection{Proof of Lemmas}
Before the proof of Lemmas \ref{lem:error-wiretap}
and \ref{lem:error-uwiretap}, we prepare the following lemmas.

\begin{lem}[{\cite[Lemma 8]{HASH}}]
 \label{lem:E}
 For any $A$ and $\uu\in\U^n$,
 \begin{equation*}
         p_C\lrsb{\lrb{\cc: A\uu=\cc}}
         =
         \sum_{c}p_C(\cc)\chi(A\uu=\cc)
         =\frac 1{|\im\A|}
 \end{equation*}
 and for any $\uu\in\U^n$
 \begin{equation*}
         E_{AC}\lrB{\chi(A\uu=\cc)}
         =\sum_{A,\cc}p_{AC}(A,\cc)\chi(A\uu=\cc)
         =\frac 1{|\im\A|}.
 \end{equation*}
\end{lem}

\begin{lem}[{\cite[Lemma 3]{HASH}}]
 \label{lem:ACnotempty}
 If $(\A,p_A)$ satisfies (\ref{eq:hash}), then
\begin{align*}
 p_{AC}\lrsb{\lrb{
 (A,\cc):
 \begin{aligned}
	&\G\cap\C_A(\cc)\neq \emptyset
	\\
	&\uu\in\C_A(\cc)
 \end{aligned}
 }}
 \leq
 \frac{|\G|\alpha_A}{|\im\A|^2} + \frac{\beta_A}{|\im\A|}
\end{align*}
 for all $\G\subset\U^n$ and all $\uu\notin\G$.
\end{lem}

\begin{lem}
 \label{lem:T}
 Assume that $\ehB>\eA>0$.
 For $\bbeta_A$ satisfying $\limn\beta_A(n)=0$ and any $\gamma>0$,
 there is a sequence $\kkappa\equiv\{\kappa(n)\}_{n=1}^{\infty}$ 
 and $\T\subset\T_{U}\subset\T_{X,\gamma}$ such that
 \begin{gather}
  \limn\kappa(n)=\infty
  \label{eq:k1}
  \\
  \limn \kappa(n)\beta_A(n)=0
  \label{eq:k2}
  \\
  \limn\frac{\log\kappa(n)}n=0
  \label{eq:k3}
 \end{gather}
 and
 \begin{gather}
  \gamma \geq D(\nu_{U}\|\mu_{X})
  \label{eq:wiretap-gamma}
  \\
  \kappa
  \leq
  \frac{|\T|}{|\im\A||\im\B||\im\hcB|}
  \leq 
  2\kappa
  \label{eq:T}
 \end{gather}
 for all sufficiently large $n$, where
 $U$ is defined as
 \begin{align}
  U\equiv\arg\min_{U'}D(\nu_{U'}\|\mu_{X}).
  \label{eq:U}
 \end{align}
 In the following, $\kappa$ denotes $\kappa(n)$.
\end{lem}

\begin{proof}
 Let
 \begin{equation}
  \kappa(n)\equiv
   \begin{cases}
    n^{\xi}
    &\text{if}\ \beta_A(n)=o\lrsb{n^{-\xi}}, \xi>0
    \\
    \frac 1{\sqrt{\beta_A(n)}},
    &\text{otherwise}
   \end{cases}
   \label{eq:kappa}
 \end{equation}
 for every $n$.
 It is clear that $\kkappa$ satisfies (\ref{eq:k1}) and (\ref{eq:k2}).
 It is also clear that $\kkappa$ satisfies (\ref{eq:k3})
 when $\beta_A(n)=o\lrsb{n^{-\xi}}$, $\xi>0$.
 If $\beta_A(n)$ is not $o\lrsb{n^{-\xi}}$,
 there is $\kappa'>0$ such that
 $\beta_A(n)n^{\xi}>\kappa'$
 and
 \begin{align}
  \frac{\log\kappa(n)}n
  &=
  \frac{\log\frac 1{\beta_A(n)}}{2n}
  \notag
  \\
  &\leq \frac{\log\frac{n^{\xi}}{\kappa'}}{2n}
  \notag
  \\
  &=\frac{\xi\log n-\log\kappa'}{2n}
 \end{align}
 for all sufficiently large $n$.
 This implies that $\kkappa$ satisfies (\ref{eq:k3}).
 The inequality (\ref{eq:wiretap-gamma}) comes from Lemma \ref{lem:vd}.
 From Lemma \ref{lem:vd} and $\ehB>\eA>0$, we have
 \begin{equation}
  \RA+\RB+\RhB+\frac{\log\kappa}n
   \leq H(U)-\lambda_{\X}
   \leq H(X)
 \end{equation}
 for all sufficiently large $n$.
 Then we have
 \begin{align}
  |\T_{X,\gamma}|
  &\geq
  |\T_{U}|
  \notag
  \\
  &\geq 2^{n[H(U)-\lambda_{\X}]},
  \notag
  \\
  &\geq
  \kappa
  2^{n[\RA+\RB+\RhB]}
  \notag
  \\
  &=
  \kappa
  |\im\A||\im\B||\im\hcB|
 \end{align}
 for all sufficiently large $n$,
 where the first inequality comes from (\ref{eq:wiretap-gamma}).
 This implies that
 there is $\T\subset\T_{U}\subset\T_{X,\gamma}$ such that
 \begin{align}
  \kappa
  &\leq
  \frac{|\T|}{|\im\A||\im\B||\im\hcB|}
  \leq 
  2\kappa
  \label{eq:wiretap-T}
 \end{align}
 for all sufficiently large $n$.
\end{proof}

\begin{rem}
 It should be noted that
 we can let $\xi$ be arbitrarily large in (\ref{eq:kappa})
 when $\beta_A(n)$ vanishes exponentially fast.
 This parameter $\xi$ affects the upper bound of
 (\ref{eq:error-wiretap-lemma}) and (\ref{eq:error-uwiretap-lemma}).
 \hfill\QED
\end{rem}

\subsection{Proof of Lemma \ref{lem:error-wiretap}}
\label{sec:lemma-wiretap}
In the following, we assume that $\eA$, $\ehB$, and $\gamma>0$ satisfy
\begin{gather}
 \ehB>\eA>
 \max\lrb{\zeta_{\X|\Y}(2\gamma|2\gamma),\zeta_{\X|\Z}(2\gamma|2\gamma)}.
 \label{eq:wiretap-ebea}
\end{gather}
Let $\kkappa\equiv\{\kappa(n)\}_{n=1}^{\infty}$ be a sequence satisfying
(\ref{eq:k1})--(\ref{eq:k3}), and $U$ be defined by (\ref{eq:U}).
Then (\ref{eq:wiretap-gamma}) is satisfied for all $\gamma>0$
and all sufficiently large $n$.
From Lemma \ref{lem:T},
there is $\T\subset\T_{U}\subset\T_{X,\gamma}$ satisfying (\ref{eq:T}).

Let $\xx$ an input of the channel,
and $\yy$ and $\zz$ be the channel outputs
of the receiver and the eavesdropper, respectively.
Let $\mm$ be a message and $\ww$ be a random sequence.
We define
\begin{align}
 &\bullet g_{AB\hB}(\cc,\mm,\ww)\in\T\subset\T_{X,\gamma}
 \tag{W1}
 \\
 &\bullet \yy\in\T_{Y|X,\gamma}(g_{AB\hB}(\cc,\mm,\ww))
 \tag{W2}
 \\
 &\bullet \zz\in\T_{Z|X,\gamma}(g_{AB\hB}(\cc,\mm,\ww))
 \tag{W3}
 \\
 &\bullet g_{A}(\cc|\yy)=g_{AB\hB}(\cc,\mm,\ww)
 \tag{W4}
 \\
 &\bullet g_{AB}(\cc,\mm|\zz)=g_{AB\hB}(\cc,\mm,\ww).
 \tag{W5}
\end{align}
Then the left hand side of (\ref{eq:error-wiretap-lemma}) is upper bounded by
\begin{align}
 \begin{split}
  &p_{MWYZ}\lrsb{\lrb{
  (\mm,\ww,\yy,\zz):
  \begin{aligned}
   &g_{AB\hB}(\cc,\mm,\ww)\notin\T_{X,\gamma}
   \\
   &\text{or}\ \zz\notin\T_{Z|X,\gamma}(g_{AB\hB}(\cc,\mm,\ww))
  \\
  &\text{or}\  g_{AB\hB}(\cc,\mm,\ww)\neq g_{A}(\cc|\yy)
  \\
  &\text{or}\  g_{AB\hB}(\cc,\mm,\ww)\neq g_{AB}(\cc,\mm|\zz)
  \end{aligned}
  }}
  \\
  &\leq
  p_{MWYZ}(\cS_1^c)
  +p_{MWYZ}(\cS_1\cap\cS_2^c)
  +p_{MWYZ}(\cS_1\cap\cS_3^c)
  +p_{MWYZ}(\cS_1\cap\cS_2\cap\cS_4^c)
  \\*
  &\quad
  +p_{MWYZ}(\cS_1\cap\cS_3\cap\cS_5^c),
 \end{split}
 \label{eq:wiretap-error0}
\end{align}
where
\begin{align*}
 \cS_i
 &\equiv
 \lrb{
 (\mm,\ww,\yy,\zz): \text{(W$i$)}
 }.
\end{align*}

First, we evaluate
$E_{AB\hB C}\lrB{p_{MWYZ}(\cS_1^c)}$.
From Lemma \ref{lem:saturating} and (\ref{eq:wiretap-T}),
we have
\begin{align}
 E_{AB\hB C}\lrB{p_{MWYZ}(\cS_1^c)}
 &=
 p_{AB\hB CMW}\lrsb{\lrb{
 (A,B,\hB, \cc,\mm,\ww):
 g_{AB\hB}(\cc,\mm,\ww)\notin\T
 }}
 \notag
 \\
 &\leq
 p_{AB\hB CMW}\lrsb{\lrb{
 (A,B,\hB,\cc,\mm,\ww):
 \T\cap\C_{AB\hB}(\cc,\mm,\ww)=\emptyset
 }}
 \notag
 \\
  &\leq
  \alpha_{AB\hB}-1
  +\frac{|\im\A||\im\B||\im\hcB|\lrB{\beta_{AB\hB}+1}}{|\T|}
 \notag
 \\
 &\leq
 \alpha_{AB\hB}-1
 +\frac{\beta_{AB\hB}+1}{\kappa}
 \notag
 \\
 &\leq
 \frac {\delta'}5
 \label{eq:wiretap-error1}
\end{align}
for all $\delta'>0$ and sufficiently large $n$,
where the last inequality comes from (\ref{eq:k1})
and the properties (\ref{eq:alpha}) and (\ref{eq:beta})
of an ensemble $(\bcA\times\bcB\times\bchB,\bp_{AB\hB})$.

Next, we evaluate
$E_{AB\hB C}\lrB{p_{MWYZ}(\cS_1\cap\cS_2^c)}$ and
$E_{AB\hB C}\lrB{p_{MWYZ}(\cS_1\cap\cS_3^c)}$.
From Lemma \ref{lem:typical-prob}, we have
\begin{align}
 E_{AB\hB C}\lrB{p_{MWYZ}(\cS_1\cap\cS_2^c)}
 &\leq
 \frac {\delta'}5
 \label{eq:w-error2}
 \\
 E_{AB\hB C}\lrB{p_{MWYZ}(\cS_1\cap\cS_3^c)}
 &\leq
 \frac {\delta'}5
 \label{eq:w-error4}
\end{align}
for all $\delta'>0$ and sufficiently large $n$.

Next,
we evaluate
$E_{AB\hB C}\lrB{p_{MWYZ}(\cS_1\cap\cS_2\cap\cS_4^c)}$ and
$E_{AB\hB C}\lrB{p_{MWYZ}(\cS_1\cap\cS_3\cap\cS_5^c)}$.
In the following, we assume that
\begin{align*}
 &\bullet \xx\in\T\subset\T_{X,\gamma}
 \\
 &\bullet \yy\in\T_{Y|X,\gamma}(\xx)
 \\
 &\bullet g_{A}(\cc|\yy)\neq\xx.
\end{align*}
From Lemma \ref{lem:typical-trans}, we have
$(\xx,\yy)\in\T_{XY,2\gamma}$ and $\xx\in\T_{X|Y,2\gamma}(\yy)$.
Then there is $\xx'\in\C_A(\cc)$ such that
$\xx'\neq\xx$ and
\begin{align*}
 \mu_{X|Y}(\xx'|\yy)
 &\geq
 \mu_{X|Y}(\xx|\yy)
 \\
 &\geq
 2^{-n[H(X|Y)+\zeta_{\X|\Y}(2\gamma|2\gamma)]},
\end{align*}
where the second inequality comes from Lemma \ref{lem:typical-prob}.
This implies that $[\G(\yy)\setminus\{\xx\}]\cap\C_A(\cc)\neq\emptyset$,
where
\[
 \G(\yy)\equiv\lrb{
 \xx': \mu_{X|Y}(\xx'|\yy)\geq 2^{-[H(X|Y)+\zeta_{\X|\Y}(2\gamma|2\gamma)]}
 }.
\]
Then we have
\begin{align}
 &E_{AB\hB C}\lrB{
 p_{MWYZ}(\cS_1\cap\cS_2\cap\cS_4^c)
 }
 \notag
 \\*
 &\leq
 E_{AB\hB CMW}\left[
 \sum_{\xx\in\T}
 \chi(g_{AB\hB}(C,M,W)=\xx)
 \sum_{\yy\in\T_{Y|X,\gamma}(\xx)}
 \mu_{Y|X}(\yy|\xx)\chi(g_{A}(C|\yy)\neq \xx)
 \right]
 \notag
 \\
 &\leq
 E_{AB\hB CMW}\left[
 \sum_{\xx\in\T}
 \chi(A\xx=C)\chi(B\xx=M)\chi(\hB\xx=W)
 \sum_{\yy\in\T_{Y|X,\gamma}(\xx)}
 \mu_{Y|X}(\yy|\xx)
 \chi(g_{A}(C|\yy)\neq \xx)
 \right]
 \notag
 \\
 &=
 \sum_{\xx\in\T}
 \sum_{\yy\in\T_{Y|X,\gamma}(\xx)}
 \mu_{Y|X}(\yy|\xx)
 E_{AC}\left[
 \chi(g_{A}(C|\yy)\neq \xx)
 \chi(A\xx=C)
 E_{B\hB MW}\lrB{
 \chi(B\xx=M)
 \chi(\hB\xx=W)
 }
 \right]
 \notag
 \\
 &=
 \frac 1{|\im\B||\im\hcB|}
 \sum_{\xx\in\T}
 \sum_{\yy\in\T_{Y|X,\gamma}(\xx)}
 \mu_{Y|X}(\yy|\xx)
 E_{AC}\left[
 \chi(g_{A}(C|\yy)\neq \xx)
 \chi(A\xx=C)
 \right]
 \notag
 \\
 &\leq
 \frac 1{|\im\B||\im\hcB|}
 \sum_{\xx\in\T}
 \sum_{\yy\in\T_{Y|X,\gamma}(\xx)}
 \mu_{Y|X}(\yy|\xx)
 p_{AC}\lrsb{\lrb{
 (A,\cc):
 \begin{aligned}
  &\lrB{\G(\yy)\setminus\{\xx\}}\cap\C_A(\cc)\neq\emptyset
  \\
  &\xx\in\C_A(\cc)
 \end{aligned}
 }}
 \notag
 \\
 &\leq
 \frac 1{|\im\B||\im\hcB|}
 \sum_{\xx\in\T}
 \sum_{\yy\in\T_{Y|X,\gamma}(\xx)}
 \mu_{Y|X}(\yy|\xx)
 \lrB{
 \frac{2^{n[H(X|Y)+\zeta_{\X|\Y}(2\gamma|2\gamma)]}\alpha_A}
 {|\im\A|^2}
 +\frac{\beta_A}{|\im\A|}
 }
 \notag
 \\
 &\leq
 \lrB{
 \frac{2^{n[H(X|Y)+\zeta_{\X|\Y}(2\gamma|2\gamma)]}\alpha_A}
 {|\im\A|}
 +\beta_A
 }
 \sum_{\xx\in\T}
 \frac 1{|\im\A||\im\B||\im\hcB|}
 \notag
 \\
 &\leq
 \frac{
 2\kappa
 |\X|^{\lA}
 2^{-n[\eA-\zeta_{\X|\Y}(2\gamma|2\gamma)]}\alpha_A
 }
 {|\im\A|}
 +2\kappa\beta_A
 \notag
 \\
 &\leq
 \frac{\delta'}5
 \label{eq:wiretap-error3}
\end{align}
for all $\delta'>0$ and sufficiently large $n$,
where the second inequality comes from Lemma~\ref{lem:E},
the fourth inequality comes from Lemma~\ref{lem:ACnotempty}
and the fact that
\[
 |\G(\yy)|\leq 2^{n[H(X|Y)+\zeta_{\X|\Y}(2\gamma|2\gamma)]},
\]
the sixth inequality comes from (\ref{eq:wiretap-T}),
and the last inequality comes from (\ref{eq:k2}), (\ref{eq:wiretap-ebea})
and the properties (\ref{eq:imA})--(\ref{eq:beta})
of an ensemble $(\bcA,\bp_A)$.
Similarly, we have
\begin{align}
 E_{AB\hB C}\lrB{
 p_{MWYZ}(\cS_1\cap\cS_3\cap\cS_5^c)
 }
 &\leq
 \frac{
 2\kappa|\X|^{\lA+\lB}
 2^{-n[\eA-\zeta_{\X|\Z}(2\gamma|2\gamma)]}\alpha_{AB}
 }
 {|\im\A||\im\B|}
 +2\kappa\beta_{AB}
 \notag
 \\
 &\leq
 \frac{\delta'}5
 \label{eq:wiretap-error5}
\end{align}
for all $\delta'>0$ and sufficiently large $n$.

Finally, from (\ref{eq:wiretap-error0})--(\ref{eq:wiretap-error5}),
we have the fact that
for all $\delta'>0$ and sufficiently large $n$
there are $A\in\A$, $B\in\B$, $\hB\in\hcB$, and $\cc\in\im\A$ such that
\begin{align*}
 p_{MWYZ}\lrsb{\lrb{
 (\mm,\ww,\yy,\zz):
 \begin{aligned}
  &g_{AB\hB}(\cc,\mm,\ww)\notin\T_{X,\gamma}
  \\
  &\text{or}\ \zz\notin\T_{Z|X,\gamma}(g_{AB\hB}(\cc,\mm,\ww))
  \\
  &\text{or}\ g_{AB\hB}(\cc,\mm,\ww)\neq g_{A}(\cc|\yy)
  \\
  &\text{or}\ g_{AB\hB}(\cc,\mm,\ww)\neq g_{AB}(\cc,\mm|\zz)
 \end{aligned}
 }}
 &\leq
 \delta'.
\end{align*}
\hfill\QED

\subsection{Proof of Lemma \ref{lem:R}}
\label{sec:lemma-uwiretap}
The following proof is based on \cite[Lemma 1]{CRYPTLDPC}.
If there is a random variable $X'$ taking values in $\X$ such
that
\begin{equation}
 H(X|X',Z)=R'
  \label{eq:HXctZR}
\end{equation}
for given $(X,Z)$ and $0\leq R'\leq H(X|Z)$
the lemma is proved by letting
\begin{align*}
 R'&\equiv H(X)-R-\e\leq H(X|Z)
 \\
 \tZ&\equiv(X',Z)
 \\
 f(\tz)&\equiv z\ \text{for}\ \tz=(x',z)
\end{align*}
because 
\begin{align}
 I(X;\tZ)
 &=H(X)-H(X|\tZ)
 \notag
 \\
 &=H(X)-R'
 \notag
 \\
 &=R+\e.
\end{align}

The following proves the existence of $X'$
satisfying (\ref{eq:HXctZR}).
It is clear that  $0\leq H(X|X',Z)\leq H(X|Z)$ 
for any $(X,X',Z)$,
$H(X|X',Z)=H(X|Z)$
when $X'$ is independent of $(X,Z)$,
and $H(X|X',Z)=0$ when $X'=X$.
Since $H(X|X',Z)$ is a continuous function
of the conditional distribution $p_{X'|XZ}$,
we have the existence of $p_{X'|XZ}$ satisfying
$H(X|X',Z)=R'$
from the intermediate value theorem,
where $p_{XX'Z}$ is given by
\[
p_{XX'Z}(x,x',v)\equiv \mu_{XZ}(x,z)p_{X'|XZ}(x'|x,z)
\]
for $(x,x',z)\in\X\times\X\times\Z$.
\hfill\qed

\subsection{Proof of Lemma \ref{lem:error-uwiretap}}
Let $\kkappa\equiv\{\kappa(n)\}_{n=1}^{\infty}$ be a sequence satisfying
(\ref{eq:k1})--(\ref{eq:k3}).
Let $U$ be defined by (\ref{eq:U}).
Then (\ref{eq:wiretap-gamma}) is satisfied for all $\gamma>0$
and all sufficiently large $n$.
From Lemma \ref{lem:T}, there is $\T\subset\T_{U}\subset\T_{X,\gamma}$
satisfying (\ref{eq:T}).

We define
\begin{align}
 &\bullet \tg_{AB}(\cc,\mm,\ww)\in\T\subset\T_U\subset\T_{X,\gamma}
 \tag{UW1}
 \\
 &\bullet \tzz\notin\T_{\tZ|X,\gamma}(\tg_{AB\hB}(\cc,\mm,\ww))
 \tag{UW2}
 \\
 &\bullet \tg_{A}(\cc|\yy)=\tg_{AB\hB}(\cc,\mm,\ww)
 \tag{UW3}
 \\
 &\bullet \tg_{AB}(\cc,\mm|\tzz)=\tg_{AB\hB}(\cc,\mm,\ww),
 \tag{UW4}
\end{align}
where we assume that $n$ is large enough to satisfy
$\T_{\tZ|X,\gamma}(\xx)\neq\emptyset$ for all $\xx\in\T_{X,\gamma}$.
Then the left hand side of (\ref{eq:error-uwiretap}) is upper bounded by
\begin{align}
 \begin{split}
  &p_{MWY\tZ}\lrsb{\lrb{
  (\mm,\ww,\yy,\tzz):
  \begin{aligned}
   &\tg_{AB}(\cc,\mm,\ww)\notin\T_{X,\gamma}
   \\
   &\text{or}\ \tzz\notin\T_{\tZ|X,\gamma}(\tg_{AB\hB}(\cc,\mm,\ww))
   \\
   &\text{or}\ \tg_{AB\hB}(\cc,\mm,\ww)\neq\tg_{A}(\cc|\yy)
   \\
   &\text{or}\ \tg_{AB\hB}(\cc,\mm,\ww)\neq\tg_{AB}(\cc,\mm|\tzz)
  \end{aligned}
 }}
 \\*
 &\leq
 p_{MWY\tZ}(\cS_1^c)
 +p_{MWY\tZ}(\cS_1\cap\cS_2^c)
 +p_{MWY\tZ}(\cS_1\cap\cS_3^c)
 +p_{MWY\tZ}(\cS_1\cap\cS_4^c),
 \end{split}
 \label{eq:uwiretap-error0}
\end{align}
where
\begin{align*}
 \cS_i
 &\equiv
 \lrb{
 (\mm,\ww,\yy,\tzz): \text{(UW$i$)}
 }.
\end{align*}

First, we evaluate 
$E_{AB\hB C}\lrB{p_{MWY\tZ}(\cS_1^c)}$.
Similarly to the proof of (\ref{eq:wiretap-error1}),
we have
\begin{align}
 E_{AB\hB C}\lrB{p_{MWY\tZ}(\cS_1^c)}
 &\leq 
 \alpha_{AB\hB}-1
 +
 \frac{\beta_{AB\hB}+1}{\kappa}.
 \label{eq:uwiretap-error1}
\end{align}

Next, we evaluate $E_{AB\hB C}\lrB{p_{MWY\tZ}(\cS_1\cap\cS_2^c)}$.
From Lemma~\ref{lem:typical-prob}, we have
\begin{align}
 E_{AB\hB C}\lrB{p_{MWY\tZ}(\cS_1\cap\cS_2^c)}
 &\leq
 2^{-n[\gamma-\lambda_{\X\tcZ}]}.
 \label{eq:uwiretap-error2}
\end{align}

Next, we evaluate
$E_{AB\hB C}\lrB{p_{MWY\tZ}(\cS_1\cap\cS_3^c)}$.
Let
\[
 \G(\yy)\equiv\{\xx': H(\xx'|\yy)\leq H(U|V)\}
\]
and assume that $(\xx,\yy)\in\T_{UV}$.
Then we have
\begin{align}
 E_{AC}\lrB{
 \chi(A\xx=C)\chi(\tg_A(C|\yy)\neq \xx)
 }
 &=
 p_{AC}\lrsb{\lrb{
 (A,\cc):
 \begin{aligned}
  &A\xx=\cc
  \\
  &\exists\xx'\neq\xx\ \text{s.t.}
  \\
  &H(\xx'|\yy)\leq H(\xx|\yy)\ \text{and}\ A\xx'=\cc
 \end{aligned}
 }}
 \notag
 \\
 &=
 p_{A}\lrsb{\lrb{
 A:
 \begin{aligned}
  &\exists\xx'\neq\xx\ \text{s.t.}
  \\
  &H(\xx'|\yy)\leq H(\xx|\yy)\ \text{and}\ A\xx'=A\xx
 \end{aligned}
 }}
 p_C\lrsb{\lrb{
 \cc: A\xx=\cc
 }}
 \notag
 \\
 &=
 \frac 1{|\im\A|}
 p_{A}\lrsb{\lrb{
 A:
 \begin{aligned}
  &\exists\xx'\neq\xx\ \text{s.t.}\ H(\xx'|\yy)\leq H(U|V)
  \\
  &\text{and}\ A\xx'=A\xx
 \end{aligned}
 }}
 \notag
 \\
 &\leq
 \frac 1{|\im\A|}
 \max\lrb{
 \sum_{\xx'\in\G(\yy)\setminus\{\xx\}}
 p_A\lrsb{\lrb{A: A\xx=A\xx'}},
 1
 }
 \notag
 \\
 &\leq
 \frac 1{|\im\A|}
 \max\lrb{
 \frac{2^{n[H(U|V)+\lambda_{\X\Y}]}\alpha_A}
 {|\im\A|}
 +\beta_A,
 1
 }
 \notag
 \\
 &=
 \frac 1{|\im\A|}
 \max\lrb{
 2^{-n[\RA-H(U|V)-\lambda_{\X\Y}]}\alpha_A
 +\beta_A,
 1
 }
 \notag
 \\
 &\leq
 \frac 1{|\im\A|}
 \lrB{
 \max\lrb{\alpha_A,1}
 2^{-n[|\RA-H(U|V)|^+
 -\lambda_{\X\Y}]}
 +\beta_A
 },
 \label{eq:uwiretap-error3sub}
\end{align}
where $|\cdot|^+$ is defined by (\ref{eq:def-plus}),
the third equality comes from Lemma \ref{lem:E},
and the second inequality comes
from Lemma \ref{lem:typeset-number} and the property (\ref{eq:hash})
of $(\A,p_A)$.
Let
\[
F_{Y|X}(R)\equiv
\min_{V|U}\lrB{D(\nu_{V|U}\|\mu_{Y|X}|\nu_U)+|R-H(U|V)|^+},
\]
where $V|U$ denotes the conditional type given type $U$.
Then we have
\begin{align}
 &E_{AB\hB C}\lrB{p_{MWY\tZ}(\cS_1\cap\cS_3^c)}
 \notag
 \\*
 &\leq
 E_{AB\hB CMW}\lrB{
 \sum_{\xx\in\T}
 \sum_{\yy}\mu_{Y|X}(\yy|\xx)
 \chi(g_{AB}(\cc,\mm,\ww)=\xx)\chi(\tg_A(\cc|\yy)\neq\xx)
 }
 \notag
 \\*
 &=
 E_{AB\hB CMW}\left[
 \sum_{\xx\in\T}\sum_{V|U}
 \sum_{\yy\in\T_{V|U}(\xx)}
 \mu_{Y|X}(\yy|\xx)
 \chi(\tg_{AB}(C,M)=\xx)\chi(\tg_A(C|\yy)\neq \xx)
 \right]
 \notag
 \\
 &\leq
 E_{AB\hB CMW}\left[
 \sum_{\xx\in\T}\sum_{V|U}
 \sum_{\yy\in\T_{V|U}(\xx)}
 \mu_{Y|X}(\yy|\xx)
 \chi(A\xx=C)\chi(B\xx=M)\chi(\hB\xx=W)
 \chi(g_A(C|\yy)\neq \xx)
 \right]
 \notag
 \\
 &=
 \sum_{\xx\in\T}
 \sum_{V|U}\sum_{\yy\in\T_{V|U}(\xx)}
 \mu_{Y|X}(\yy|\xx)
 E_{AC}\left[
 \chi(A\xx=C)\chi(g_A(C|\yy)\neq \xx)
 \right]
 E_{B\hB MW}\lrB{
 \chi(B\xx=M)\chi(\hB\xx=W)
 }
 \notag
 \\
 &
 \leq
 \frac 1{|\im\A||\im\B||\im\hcB|}
 \sum_{\xx\in\T}
 \sum_{V|U}
 \sum_{\yy\in\T_{V|U}(\xx)}
 \mu_{Y|X}(\yy|\xx)
 \lrB{
 \max\lrb{\alpha_A,1}
 2^{-n[|\RA-H(U|V)|^+
 -\lambda_{\X\Y}]}
 +\beta_A
 }
 \notag
 \\
 &
 =
 \frac 1{|\im\A||\im\B||\im\hcB|}
 \sum_{\xx\in\T}
 \lrB{
 \max\lrb{\alpha_A,1}
 \sum_{V|U}
 \sum_{\yy\in\T_{V|U}(\xx)}
 \mu_{Y|X}(\yy|\xx)
 2^{-n[|\RA-H(U|V)|^+
 -\lambda_{\X\Y}]}
 +\beta_A
 }
 \notag
 \\
 &
 \leq
 \frac 1{|\im\A||\im\B||\im\hcB|}
 \sum_{\xx\in\T}
 \lrB{
 \max\lrb{\alpha_A,1}
 \sum_{V|U}
 2^{-n[D(\nu_{V|U}\|\mu_{Y|X}|\nu_U)+|\RA-H(U|V)|^+
 -\lambda_{\X\Y}]}
 +\beta_A
 }
 \notag
 \\
 &
 \leq
 \frac {|\T|}{|\im\A||\im\B||\im\hcB|}
 \lrB{
 \max\lrb{\alpha_A,1}
 2^{-n[F_{Y|X}(\RA)-2\lambda_{\X\Y}]}
 +\beta_A
 }
 \notag
 \\
 &\leq
 2\kappa
 \lrB{
 \max\lrb{\alpha_A,1}
 2^{-n[F_{Y|X}(\RA)-2\lambda_{\X\Y}]}
 +\beta_A
 },
 \label{eq:uwiretap-error3}
\end{align}
where the third inequality comes from
Lemma \ref{lem:E} and (\ref{eq:uwiretap-error3sub}),
the fourth inequality comes from Lemmas~\ref{lem:exprob}
 and~\ref{lem:typenumber}, the fifth inequality comes from
the definition of $F_{Y|X}$ and Lemma~\ref{lem:typebound},
and the last inequality comes from (\ref{eq:wiretap-T}).
Similarly, we have
\begin{align}
 E_{ABC}\lrB{p_{MWY\tZ}(\cS_1\cap\cS_4^c)}
 &\leq
 2\kappa
 \lrB{
 \max\lrb{\alpha_{AB},1}
 2^{-n[F_{\tZ|X}(\RA+\RB)-2\lambda_{\X\tcZ}]}
 +\beta_{AB}
 },
 \label{eq:uwiretap-error4}
\end{align}
where
\[
F_{\tZ|X}(R)\equiv
\min_{V'|U}\lrB{D(\nu_{V'|U}\|\mu_{\tZ|X}|\nu_U)+|R-H(U|V')|^+}.
\]

From (\ref{eq:uwiretap-error0})--(\ref{eq:uwiretap-error2}),
(\ref{eq:uwiretap-error3}) and (\ref{eq:uwiretap-error4}),
we have
\begin{align*}
 &E_{AB\hB C}\lrB{
 p_{MWY\tZ}\lrsb{\lrb{
 (\mm,\ww,\yy,\tzz):
 \begin{aligned}
  &\tg_{AB}(\cc,\mm,\ww)\notin\T_{X,\gamma}
  \\
  &\text{or}\ \tzz\notin\T_{\tZ|X,\gamma}(\tg_{AB\hB}(\cc,\mm,\ww))
  \\
  &\text{or}\ \tg_{AB\hB}(\cc,\mm,\ww)\neq\tg_{A}(\cc|\yy)
  \\
  &\text{or}\ \tg_{AB\hB}(\cc,\mm,\ww)\neq\tg_{AB}(\cc,\mm|\tzz)
 \end{aligned}
 }}
 }
 \notag
 \\*
 &\leq
 \alpha_{AB\hB}-1
 +\frac{\beta_{AB\hB}+1}{\kappa}
 +2^{-n[\gamma-\lambda_{\X\tcZ}]}
 +2\kappa
 \lrB{
 \max\lrb{\alpha_A,1}
 2^{-n[\inf F_{Y|X}(\RA)-2\lambda_{\X\Y}]}
 +\beta_A
  }
 \\*
 &\quad
 +
 2\kappa
 \lrB{
 \max\lrb{\alpha_{AB},1}
 2^{-n[\inf F_{\tZ|X}(\RA+\RB)-2\lambda_{\X\tcZ}]}
 +\beta_{AB}
 },
\end{align*}
where the infimum is taken over all
$\mu_{Y\tZ|X}$ satisfying
(\ref{eq:RABhB-uwiretap-lemma})--(\ref{eq:RAB-uwiretap-lemma}).
This  implies that there are $A\in\A$, $B\in\B$, $\hB\in\hcB$, and
$\cc\in\im\A$ such that
\begin{align}
 \begin{split}
 &p_{MWY\tZ}\lrsb{\lrb{
  (\mm,\ww,\yy,\tzz):
  \begin{aligned}
   &\tg_{AB}(\cc,\mm,\ww)\notin\T_{X,\gamma}
  \\
   &\text{or}\ \tzz\notin\T_{\tZ|X,\gamma}(\tg_{AB\hB}(\cc,\mm,\ww))
   \\
   &\text{or}\ \tg_{AB\hB}(\cc,\mm,\ww)\neq\tg_{A}(\cc|\yy)
   \\
   &\text{or}\ \tg_{AB\hB}(\cc,\mm,\ww)\neq\tg_{AB}(\cc,\mm|\tzz)
  \end{aligned}
  }}
  \\*
  &\leq
  \alpha_{AB\hB}-1
  +\frac{\beta_{AB\hB}+1}{\kappa}
  +2^{-n[\gamma-\lambda_{\X\tcZ}]}
  +2\kappa
  \lrB{
  \max\lrb{\alpha_A,1}
  2^{-n[\inf F_{Y|X}(\RA)-2\lambda_{\X\Y}]}
  +\beta_A
  }
  \\*
  &\quad
  +2\kappa
  \lrB{
  \max\lrb{\alpha_{AB},1}
  2^{-n[\inf F_{\tZ|X}(\RA+\RB)-2\lambda_{\X\tcZ}]}
  +\beta_{AB}
  }.
 \end{split}
 \label{eq:error-uwiretap-lemma-proof}
\end{align}
Since
\begin{gather*}
 \inf_{\substack{
 \mu_{Y|X}:\\
 H(Y|X)<\RA
 }}
 F_{Y|X}(\RA)>0
 \\
 \inf_{\substack{
 \mu_{\tZ|X}:\\
 H(\tZ|X)<\RA+\RB
 }}
 F_{\tZ|X}(\RA+\RB)>0,
\end{gather*}
then the right hand side of (\ref{eq:error-uwiretap-lemma-proof}) goes
to zero as $n\to\infty$
by assuming (\ref{eq:k1})--(\ref{eq:k3}) and
the properties (\ref{eq:alpha}) and (\ref{eq:beta})
of ensembles $(\bcA,\bp_A)$, $(\bcA\times\bcB,\bp_{AB})$ and
$(\bcA\times\bcB\times\bchB,\bp_{AB\hB})$.
\hfill\QED

\subsection{Proof of Lemma \ref{lem:error-ska}}
\label{sec:lemma-ska}

In the following, we assume that $\eA>0$ and $\gamma>0$ satisfy
\begin{gather}
 \eA
 >\max\lrb{\zeta_{\X|\Y}(\gamma|\gamma),\zeta_{\X|\Z}(\gamma|\gamma)}.
 \label{eq:ska-eA}
\end{gather}

Let $\xx$, $\yy$, and $\zz$ be outputs of correlated sources.
We define
\begin{align}
 &\bullet (\xx,\yy,\zz)\in\T_{XYZ,\gamma}
 \tag{SKA1}
 \\
 &\bullet g_{A}(A\xx|\yy)=\xx
 \tag{SKA2}
 \\
 &\bullet g_{AB}(A\xx,B\xx|\zz)=\xx.
 \tag{SKA3}
\end{align}
Then the left hand side of (\ref{eq:error-ska}) is upper bounded by
\begin{align}
 \begin{split}
  p_{XYZ}\lrsb{\lrb{
  (\xx,\yy,\zz):
  \begin{aligned}
   &g_{AB}(A\xx|\yy)\neq \xx
   \\
   &\text{or}\ g_{AB\hB}(A\xx,B\xx|\zz)\neq \xx
  \end{aligned}
  }}
  &\leq
  p_{XYZ}(\cS_1^c)
  +p_{XYZ}(\cS_1\cap\cS_2^c)
  +p_{XYZ}(\cS_1\cap\cS_3^c),
 \end{split}
 \label{eq:ska-error0}
\end{align}
where
\begin{align*}
 \cS_i
 &\equiv
 \lrb{
 (\xx,\yy,\zz): \text{(SKA$i$)}
 }.
\end{align*}

First, we evaluate
$E_{AB}\lrB{p_{XYZ}(\cS_1^c)}$.
From (\ref{eq:wiretap-T}),
we have
\begin{align}
 E_{AB}\lrB{p_{XYZ}(\cS_1^c)}
 &\leq
 2^{-n[\gamma-\lambda_{\X\Y\Z}]}
 \notag
 \\
 &\leq
 \frac {\delta'}3
 \label{eq:ska-error1}
\end{align}
for all $\delta'>0$ and sufficiently large $n$.

Next, we evaluate
$E_{AB}\lrB{p_{XYZ}(\cS_1\cap\cS_2^c)}$ and
$E_{AB}\lrB{p_{XYZ}(\cS_1\cap\cS_3^c)}$.
From Lemma \ref{lem:typical-trans}, we have
$(\xx,\yy)\in\T_{XY,\gamma}$ and $\xx\in\T_{X|Y,\gamma}(\yy)$.
Then there is $\xx'\in\C_A(A\xx)$ such that
$\xx'\neq\xx$ and
\begin{align}
 \mu_{X|Y}(\xx'|\yy)
 &\geq
 \mu_{X|Y}(\xx|\yy)
 \notag
 \\
 &\geq
 2^{-n[H(X|Y)+\zeta_{\X|\Y}(\gamma|\gamma)]},
\end{align}
where the second inequality comes from Lemma \ref{lem:typical-prob}.
This implies that $[\G(\yy)\setminus\{\xx\}]\cap\C_A(A\xx)\neq\emptyset$,
where
\[
 \G(\yy)\equiv\lrb{
 \xx': \mu_{X|Y}(\xx'|\yy)\geq 2^{-[H(X|Y)+\zeta_{\X|\Y}(\gamma|\gamma)]}
 }.
\]
Then we have
\begin{align}
 E_{AB}\lrB{\mu_{XYZ}(\cS_1\cap\cS_2^c)}
 &\leq
 \sum_{(\xx,\yy,\zz)\in\T_{XYZ,\gamma}}
 \mu_{XYZ}(\xx,\yy,\zz)
 p_{A}\lrsb{\lrb{
 A:
 \lrB{\G(\yy)\setminus\{\xx\}}\cap\C_{A}(A\xx)
 \neq\emptyset
 }}
 \notag
 \\
 &\leq
 \sum_{(\xx,\yy,\zz)\in\T_{XYZ,\gamma}}
 \mu_{XYZ}(\xx,\yy,\zz)
 \lrB{
 \frac{|\G(\yy)|\alpha_A}{|\im\A|}
 +\beta_A
 }
 \notag
 \\
 &\leq
 \sum_{(\xx,\yy,\zz)\in\T_{XYZ,\gamma}}
 \mu_{XYZ}(\xx,\yy,\zz)
 \lrB{
 \frac{2^{n[H(X|Y)+\zeta_{\X|\Y}(\gamma|\gamma)]}\alpha_A}{|\im\A|}
 +\beta_A
 }
 \notag
 \\
 &\leq
 \frac{|\X|^{\lA}\alpha_A}{|\im\A|}
 2^{n[H(X|Y)+\zeta_{\X|\Y}(\gamma|\gamma)]}|\X|^{-\lA}
 +\beta_A
 \notag
 \\
 &\leq
 \frac{|\X|^{\lA}\alpha_A}{|\im\A|}
 2^{-n[\eA-\zeta_{\X|\Y}(\gamma|\gamma)]}
 +\beta_A
 \notag
 \\
 &\leq\frac{\delta'}3
 \label{eq:ska-error2}
\end{align}
for all $\delta'>0$ and sufficiently large $n$ by taking an appropriate
$\gamma>0$,
where the second inequality comes from Lemma~\ref{lem:Anotempty}
and the third inequality comes from the fact that
\[
 |\G(\yy)|\leq 2^{n[H(X|Y)+\zeta_{\X|\Y}(\gamma|\gamma)]},
\]
the fifth inequality comes from the definition of $\lA$,
and the last inequality comes from 
(\ref{eq:ska-eA}) and
the properties (\ref{eq:imA})--(\ref{eq:beta}) of an ensemble $(\bcA,\bp_A)$.
Similarly, we have
\begin{align}
 E_{AB}\lrB{
 p_{XYZ}(\cS_1\cap\cS_3^c)
 }
 &
 \leq
 \frac{|\X|^{\lA+\lB}\alpha_{AB}}{|\im\A||\im\B|}
 2^{-n[\eA-\zeta_{\X|\Z}(\gamma|\gamma)]}
 +\beta_{AB}
 \notag
 \\
 &\leq\frac{\delta'}3
 \label{eq:ska-error3}
\end{align}
for all $\delta'>0$ and sufficiently large $n$.

Finally, from (\ref{eq:ska-error0})--(\ref{eq:ska-error3}),
we have the fact that
for all $\delta'>0$ and sufficiently large $n$
there are $A\in\A$ and $B\in\B$ such that
\begin{align*}
 p_{XYZ}\lrsb{\lrb{
 (\xx,\yy,\zz):
 \begin{aligned}
  &g_{A}(A\xx|\yy)\neq \xx
  \\
  &\text{or}\ g_{AB}(A\xx,B\xx|\zz)\neq \xx
 \end{aligned}
 }}
 \leq
 \delta'
\end{align*}
for all $\delta'>0$ and sufficiently large $n$.
\hfill\QED

\begin{rem}
It should be noted that
the property (\ref{eq:alpha}) of ensembles $(\bcA,\bp_A)$ and
 $(\bcA\times\bcB,\bp_{AB})$
can be replaced by
\begin{align*}
 &\limn\frac{\log\alpha_A(n)}n=1
 \\
 &\limn\frac{\log\alpha_{AB}(n)}n=1,
\end{align*}
respectively.
In particular, there are expurgated ensembles $(\bcA,\bp_A)$ and $(\bcB,\bp_B)$
of sparse matrices that have an $(\aalpha_A,\zero)$-hash property,
where the condition (\ref{eq:alpha}) for $\aalpha_A$ and $\aalpha_{AB}$
is replaced by the above respective conditions
(see \cite{BB}).
\end{rem}

\subsection{Proof of Lemma \ref{lem:R-uska}}
It has already been proved in the proof of Lemma \ref{lem:R}
that there is a random variable $X'$ taking values in $\X$ such
that
\begin{equation*}
 H(X|X',Z)=R'
\end{equation*}
for given $(X,Z)$ and $0\leq R'\leq H(X|Z)$.
The lemma is proved by letting
\begin{align*}
 R'&\equiv R-\e
 \\
 \tZ&\equiv(X',Z)
 \\
 f(\tz)&\equiv z\ \text{for}\ \tz=(x',z).
\end{align*}
\hfill\QED

\subsection{Proof of Lemma \ref{lem:error-uska}}

Let $\xx$, $\yy$, $\tzz$ be outputs of the correlated sources.
We define
\begin{align}
 &\bullet \tg_{A}(A\xx|\yy)=\xx
 \tag{USKA1}
 \\
 &\bullet \tg_{AB}(A\xx,A\xx|\tzz)=\xx.
 \tag{USKA2}
\end{align}
Then the left hand side of (\ref{eq:error-uska-lemma}) is upper bounded by
\begin{equation}
 p_{XY\tZ}\lrsb{\lrb{
 (\xx,\yy,\tzz):
 \begin{aligned}
  &g_{AB}(A\xx|\yy)\neq \xx
  \\
  &\text{or}\ g_{AB\hB}(A\xx,B\xx|\tzz)\neq \xx
 \end{aligned}
 }}
 \leq
 p_{XYZ}(\cS_1^c)+p_{XYZ}(cS_2^c),
 \label{eq:uska-error0}
\end{equation}
where
\begin{align*}
 \cS_i
 &\equiv
 \lrb{
 (\xx,\yy,\tzz): \text{(USKA$i$)}
 }.
\end{align*}

In the following, we evaluate
$E_{AB}\lrB{p_{XY\tZ}(\cS_1^c)}$
and 
$E_{AB}\lrB{p_{XY\tZ}(\cS_2^c)}$.
Let $UV$ be the type of sequence $(\xx,\yy)\in\X^n\times\Y^n$
and $V|U$ be the conditional type given type $U$.
In the following, we assume that $(\xx,\yy)\in\T_{UV}$.
If $\tg_{A}(A\xx|\yy)\neq \xx$,
then there is $\xx'\in\C_A(A\xx)$ such that
$\xx'\neq\xx$ and
\begin{align*}
 H(\xx'|\yy)
 \leq
 H(\xx|\yy)
 \leq
 H(U|V).
\end{align*}
This implies that $[\G(\yy)\setminus\{\xx\}]\cap\C_A(A\xx)\neq\emptyset$, where
\[
 \G(\yy)\equiv\lrb{
 \xx': H(\xx'|\yy)\leq H(U|V)
 }.
\]
Then we have
\begin{align}
 E_{AB}\lrB{\chi(\tg_{A}(A\xx|\yy)\neq \xx)}
 &\leq
 p_{A}\lrsb{\lrb{
 A:
 \lrB{\G(\yy)\setminus\{\xx\}}\cap\C_{A}(A\xx)
 \neq\emptyset
 }}
 \notag
 \\
 &\leq
 \max\lrb{
 \frac{|\G(\yy)|\alpha_A}{|\im\A|}+\beta_A,
 1
 }
 \notag
 \\
 &\leq
 \max\lrb{
 \frac{2^{n[H(U|V)+\lambda_{\X\Y}]}\alpha_A}{|\im\A|}+\beta_A,
 1}
 \notag
 \\
 &\leq
 \max\lrb{
 \frac{|\X|^{\lA}\alpha_A}{|\im\A|},
 1
 }
 2^{-n[|\RA-H(U|V)|^+ -\lambda_{\X\Y}]}
 +\beta_A,
 \label{eq:uska-error1sub}
\end{align}
where $|\cdot|^+$ is defined by (\ref{eq:def-plus}),
the second inequality comes from Lemma~\ref{lem:Anotempty},
and the third inequality comes from Lemma \ref{lem:typeset-number}.
Let
\[
 F_{XY}(R)\equiv\min_{UV}\lrB{D(\nu_{\xx\yy}\|\mu_{XY})+|R-H(U|V)|^+}.
\]
Then we have
\begin{align}
 E_{AB}\lrB{\mu_{XY\tZ}(\cS_1^c)}
 &=
 \sum_{UV}
 \sum_{(\xx,\yy)\in\T_{UV}}
 \mu_{XY}(\xx,\yy)
 E_{AB}\lrB{
 \chi(\tg_A(A\xx|\yy)\neq\xx)
 }
 \notag
 \\
 &\leq
 \sum_{UV}
 \sum_{(\xx,\yy)\in\T_{UV}}
 \mu_{XY}(\xx,\yy)
 \lrB{
 \max\lrb{
 \frac{|\X|^{\lA}\alpha_A}{|\im\A|},
 1
 }
 2^{-n[|\RA-H(U|V)|^+ -\lambda_{\X\Y}]}
 +\beta_A
 }
 \notag
 \\
 &\leq
 \max\lrb{
 \frac{|\X|^{\lA}\alpha_A}{|\im\A|},
 1
 }
 \sum_{UV}
 2^{-n[D(\nu_{\xx\yy}\|\mu_{XY})+|\RA-H(U|V)|^+ -\lambda_{\X\Y}]}
 +\beta_A
 \notag
 \\
 &\leq
 \max\lrb{
 \frac{|\X|^{\lA}\alpha_A}{|\im\A|},
 1
 }
 2^{-n[F_{XY}(\RA)-2\lambda_{\X\Y}]}
 +\beta_A,
 \label{eq:uska-error1}
\end{align}
where
the first inequality comes from (\ref{eq:uska-error1sub}),
the second inequality comes from 
Lemmas \ref{lem:exprob} and \ref{lem:typenumber},
and the last inequality comes from Lemma \ref{lem:typebound}
and the definition of $F_{XY}$.
Similarly, we have
\begin{align}
 E_{AB}\lrB{
 p_{XY\tZ}(\cS_2^c)
 }
 &\leq
 \max\lrb{
 \frac{|\X|^{\lA+\lB}\alpha_{AB}}{|\im\A||\im\B|},
 1
 }
 2^{-n[F_{X\tZ}(\RA+\RB)-2\lambda_{\X\tcZ}]}
 +\beta_{AB},
 \label{eq:uska-error2}
\end{align}
where
\[
 F_{X\tZ}(R)\equiv\min_{UV'}\lrB{D(\nu_{\xx\tzz}\|\mu_{X\tZ})+|R-H(U|V')|^+}.
\]

Finally, from (\ref{eq:uska-error0}), (\ref{eq:uska-error1}), and
(\ref{eq:uska-error2}),
we have
\begin{align*}
 &E_{AB}\lrB{
 p_{XY\tZ}\lrsb{\lrb{
 (\xx,\yy,\tzz):
 \begin{aligned}
  &g_{AB}(A\xx|\yy)\neq \xx
  \\
  &\text{or}\ g_{AB\hB}(A\xx,B\xx|\tzz)\neq \xx
 \end{aligned}
 }}
 }
 \notag
 \\*
 &\leq
 \max\lrb{
 \frac{|\X|^{\lA}\alpha_A}{|\im\A|},
 1
 }
 2^{-n[\inf F_{XY}(\RA)-2\lambda_{\X\Y}]}
 +
 \max\lrb{
 \frac{|\X|^{\lA+\lB}\alpha_{AB}}{|\im\A||\im\B|},
 1
 }
 2^{-n[\inf F_{X\tZ}(\RA+\RB)-2\lambda_{\X\tcZ}]}
 \\*
 &\quad
 +\beta_A
 +\beta_{AB},
\end{align*}
where the infimum is taken over all $\mu_{XY\tZ}$ satisfying
(\ref{eq:RA-uska-lemma}) and (\ref{eq:RAB-uska-lemma}).
This  implies that there are $A\in\A$ and $B\in\B$ such that
\begin{align}
 &
 p_{XY\tZ}\lrsb{\lrb{
 (\xx,\yy,\tzz):
 \begin{aligned}
  &g_{AB}(A\xx|\yy)\neq \xx
  \\
  &\text{or}\ g_{AB\hB}(A\xx,B\xx|\tzz)\neq \xx
 \end{aligned}
 }}
 \notag
 \\*
 &\leq
 \max\lrb{
 \frac{|\X|^{\lA}\alpha_A}{|\im\A|},
 1
 }
 2^{-n[\inf F_{XY}(\RA)-2\lambda_{\X\Y}]}
 +
 \max\lrb{
 \frac{|\X|^{\lA+\lB}\alpha_{AB}}{|\im\A||\im\B|},
 1
 }
 2^{-n[\inf F_{X\tZ}(\RA+\RB)-2\lambda_{\X\tcZ}]}
 \notag
 \\*
 &\quad
 +\beta_A
 +\beta_{AB}.
 \label{eq:error-uska-lemma-proof}
\end{align}
Since
\begin{gather*}
 \inf_{\substack{
 \mu_{XY}:\\
 H(X|Y)<\RA
 }}
 F_{XY}(\RA)>0
 \\
 \inf_{\substack{
 \mu_{X\tZ}:\\
 H(X|\tZ)<\RA+\RB
 }}
 F_{X\tZ}(\RA+\RB)>0,
\end{gather*}
then the right hand side of (\ref{eq:error-uska-lemma-proof}) goes to
zero as $n\to\infty$
by assuming the properties (\ref{eq:imA})--(\ref{eq:beta})
of $(\bcA,\bp_A)$ and $(\bcA\times\bcB,\bp_{AB})$.
\hfill\QED

\subsection{Method of Types}
\label{sec:type-theory}

We use the following lemmas for a set of typical sequences.

\begin{lem}[{\cite[Lemma 2.6]{CK}\cite[Lemma 21]{HASH}}]
 \label{lem:exprob}
 \begin{align*}
  \frac 1n\log \frac 1{\mu_{UV}(\uu,\vv)}
  &= H(\nu_{\uu\vv})+D(\nu_{\uu\vv}\|\mu_{UV})
  \\
  \frac 1n\log\frac 1{\mu_{U|V}(\uu|\vv)}
  &= H(\nu_{\uu|\vv}|\nu_{\vv})
  +D(\nu_{\uu|\vv}\|\mu_{U|V}|\nu_{\vv}).
 \end{align*}
\end{lem}

\begin{lem}[{\cite[Theorem 2.5]{UYE}\cite[Lemma 22]{HASH}}]
 \label{lem:typical-trans}
 If $\vv\in\T_{V,\gamma}$ and $\uu\in\T_{U|V,\gamma'}(\vv)$,
 then $(\uu,\vv)\in\T_{UV,\gamma+\gamma'}$.
 If $(\uu,\vv)\in\T_{UV,\gamma}$, then $\uu\in\T_{U,\gamma}$
 and $\uu\in\T_{U|V,\gamma}(\vv)$.
\end{lem}

\begin{lem}[{\cite[Theorem 2.7]{UYE}\cite[Lemma 24]{HASH}}]
 \label{lem:typical-aep}
 Let $0<\gamma\leq 1/8$.
 Then,
 \begin{align*}
  \begin{split}
   \left|
   \frac 1{n}\log\frac 1{\mu_{U}(\uu)} - H(U)
   \right|
   &\leq
   \zeta_{\U}(\gamma)
  \end{split}
 \end{align*}
 for all $\uu\in\T_{U,\gamma}$,
 and
 \begin{align*}
  \begin{split}
   \left|
   \frac 1{n}\log\frac 1{\mu_{U|V}(\uu|\vv)} - H(U|V)
   \right|
   &\leq
   \zeta_{\U|\V}(\gamma'|\gamma)
  \end{split}
 \end{align*}
 for $\vv\in\T_{V,\gamma}$ and $\uu\in\T_{U|V,\gamma'}(\vv)$,
 where $\zeta_{\U}(\gamma)$ and $\zeta_{\U|\V}(\gamma'|\gamma)$
 are defined in (\ref{eq:zeta}) and (\ref{eq:zetac}), respectively.
\end{lem}

\begin{lem}[{\cite[Theorem 2.8]{UYE}\cite[Lemma 25]{HASH}}]
 \label{lem:typical-prob}
 For any $\gamma>0$ and $\vv\in\V^n$,
 \begin{align*}
  \mu_U([\T_{U,\gamma}]^c)
  &\leq
  2^{-n[\gamma-\lambda_{\U}]}
  \\
  \mu_{U|V}([\T_{U|V,\gamma}(\vv)]^c|\vv)
  &\leq
  2^{-n[\gamma-\lambda_{\U\V}]},
 \end{align*}
 where $\lambda_{\U}$ and $\lambda_{\U\V}$ are defined in (\ref{eq:lambda}).
\end{lem}

\begin{lem}[{\cite[Theorem 2.9]{UYE}\cite[Lemma 26]{HASH}}]
 \label{lem:typical-number}
 For any $\gamma>0$,
 \begin{align*}
  \left|
  \frac 1{n}\log |\T_{U,\gamma}| - H(U)
  \right|
  &\leq
  \eta_{\U}(\gamma),
 \end{align*}
 where $\eta_{\U}(\gamma)$ is defined in (\ref{eq:def-eta}).
\end{lem}

\begin{lem}[{\cite[Lemma 2.2]{CK}}]
\label{lem:typebound}
 The number of different types of sequences in $\U^n$
 is fewer than $[n+1]^{|\U|}$.
 The number of conditional types of sequences in $\U^n\times\V^n$
 is fewer than $[n+1]^{|\U||\V|}$.
\end{lem}

\begin{lem}[{\cite[Lemma 2.3]{CK}}]
 \label{lem:typenumber}
 For a type $U$ of a sequence in $\X^n$,
 \begin{align*}
  2^{n[H(U)-\lambda_{\X}]}\leq |\T_U|\leq 2^{nH(U)},
 \end{align*}
 where $\lambda_{X}$ is defined in (\ref{eq:lambda}).
\end{lem}

\begin{lem}[{\cite[Lemma 7]{HASH-UNIV}\cite[Lemma 2]{CRYPTLDPC}}]
\label{lem:typeset-number}
 For $\yy\in\T_V$,
 \begin{align*}
  |\lrb{\xx': H(\xx')\leq H(U)}|
  &\leq
  2^{n[H(U)+\lambda_{\X}]}
  \\
  |\lrb{\xx': H(\xx'|\yy)\leq H(U|V)}|
  &\leq
  2^{n[H(U|V)+\lambda_{\X\Y}]},
 \end{align*}
 where $\lambda_{X}$ and $\lambda_{\X\Y}$ are defined in (\ref{eq:lambda}).
\end{lem}

\begin{lem}[{\cite[Lemma 7]{HASH-UNIV}}]
 \label{lem:vd}
 For any probability distribution $\mu_X$ on $\X$,
 \begin{gather*}
  \min_{U}d(\nu_U,\mu_X)\leq \frac{|\X|}{n}
  \\
  \min_{U}D(\nu_U\|\mu_X)\leq 
  \frac{|\X|}{n\min_{x:\mu_X(x)>0}\mu_X(x)}
  \\
  \min_{U}\left|H(X)-H(U)\right|\leq 
  \frac{2|\X|}{n\min_{x:\mu_X(x)>0}\mu_X(x)}
 \end{gather*}
 where minimum is taken over all types $U$ of the sequence in $\X^n$.
\end{lem}

\section*{Acknowledgements}
This paper was written while one of authors J.\ M.\ was a visiting
researcher at ETH, Z\"urich.
He wishes to thank Prof.\ Maurer for arranging for his stay.
The authors wish to thank
Prof.\ Uyematsu, Prof.\ Matsumoto, and Prof.\ Watanabe for helpful discussions.

\end{document}